\documentclass [journal,onecolumn,11pt]{IEEEtran}
\usepackage{amsfonts,amsmath,amssymb}
\usepackage{indentfirst, setspace}
\usepackage{url,float}
\usepackage{color}
\usepackage[a4paper,ignoreall]{geometry}
\usepackage{longtable}
\usepackage{graphicx}
\usepackage[a4paper,ignoreall]{geometry}
\usepackage{multicol}
\usepackage{stfloats}
\usepackage{enumerate}
\usepackage{cite}
\usepackage{amsthm}
\usepackage{booktabs}
\usepackage{mathrsfs}
\usepackage{multirow}
\usepackage{amssymb}
\usepackage{verbatim}
\usepackage{cases}
\usepackage[square, comma, sort&compress, numbers]{natbib}
\usepackage[overload]{empheq}
\def\qu#1 {\fbox {\footnote {\ }}\ \footnotetext { From Qu: {\color{red}#1}}}
\def\hqu#1 {}
\def\kq#1 {\fbox {\footnote {\ }}\ \footnotetext { From KangQuan: {\color{blue}#1}}}
\def\hkq#1 {}

\newcommand{\mn}{\color{black}}

\usepackage{stfloats}

\newtheorem{Th}{Theorem}[section]

\newtheorem{Prop}[Th]{Proposition}
\newtheorem{Prob}[Th]{Problem}
\newtheorem{Lemma}[Th]{Lemma}

\newtheorem{Rem}[Th]{Remark}

\newcommand{\tr}{{\rm Tr}}

\newcommand{\gf}{{\mathbb F}}

\makeatletter
\newcommand{\figcaption}{\def\@captype{figure}\caption}
\newcommand{\tabcaption}{\def\@captype{table}\caption}
\makeatother

\begin{document}
	\title{A complete characterization of the APN property \\ of a class of quadrinomials}
	\author{{ Kangquan Li, Chunlei Li, Tor Helleseth and Longjiang Qu}
	\thanks{\noindent	Kangquan Li and Longjiang Qu are with the College of Liberal Arts and Sciences,
		National University of Defense Technology, Changsha, 410073, China.
		Chunlei Li and Tor Helleseth are with the Department of Informatics, University of Bergen, Bergen N-5020, Norway.
		Longjiang Qu is also with
		the State Key Laboratory of Cryptology, Beijing, 100878, China. The work of Longjiang Qu was supported by the Nature Science Foundation of China (NSFC) under Grant  61722213, 11531002,   National Key R$\&$D Program of China (No.2017YFB0802000),  and the Open Foundation of State Key Laboratory of Cryptology. The work of Tor Helleseth and Chunlei Li was supported by the Research Council of Norway (No.~247742/O70 and No.~311646/O70).
		The work of Chunlei Li was also supported in part by the National Natural Science Foundation of China under Grant (No.~61771021). The work of Kangquan Li was supported by China Scholarship Council. Longjiang Qu is the corresponding author.
		\smallskip
		\textbf{Emails}: 			likangquan11@nudt.edu.cn,  chunlei.li@uib.no, Tor.Helleseth@uib.no, 
		ljqu\_happy@hotmail.com
	}
}
	\maketitle{}

\begin{abstract}
In this paper, by the Hasse-Weil bound, we determine the necessary and sufficient condition on  coefficients $a_1,a_2,a_3\in\gf_{2^n}$ with $n=2m$ such that $f(x) = {x}^{3\cdot2^m} + a_1x^{2^{m+1}+1} + a_2 x^{2^m+2} + a_3x^3$ is an APN function over $\gf_{2^n}$. Our result resolves the first half of an open problem by Carlet in International Workshop on the Arithmetic of Finite Fields, 83-107, 2014.
\end{abstract}


\section{Introduction}

Substitution boxes, known as S-boxes, are crucial nonlinear building blocks in modern block ciphers. 
Sboxes used in block ciphers are required to satisfy a variety of cryptographic criteria in accordance with known attacks. For instance, the differential uniformity of an S-box, one of the most classic and the most important properties, characterizes the resistance of the cryptographic component against differential cryptanalysis \cite{biham1991differential}. Moreover, it is well known that for even characteristic, the almost perfect nonlinear (APN for short) functions with differential uniformity $2$ provide the best resistance to the differential attacks. In the last three decades,  APN functions have been extensively studied  and  
the construction of infinite classes of APN functions (even permutations) is one of the most important topics \cite{carlet-2010}. 
Up to now, there are $6$ (resp. $10$) infinite classes of APN power (resp. quadratic) functions \cite[Ch. 11]{carlet-2020}.

Let $n=2m$ and $\gf_{2^{n}}$ be the finite field with $2^n$ elements.
Denote $\bar{x} = x^{2^m}$ for any $x\in\gf_{2^n}$. In this paper, by the Hasse-Weil bound, we completely characterize the APN property of the following quadrinomial 
\begin{equation}
\label{f(x)}f(x) = \bar{x}^3 + a_1\bar{x}^2x+a_2\bar{x}x^2+a_3x^3,
\end{equation}
where  $a_1,a_2,a_3\in\gf_{2^n}$. 
The motivation of studying this quadrinomial originates from 
the Kim function $\kappa(x)=x^3+x^{10}+ux^{24}$, where $u$ is a root in $\gf_{2^6}$ of the primitive polynomial $x^6+x^4+x^3+x+1$,
which is Carlet-Charpin-Zinoviev (CCZ) equivalent \cite{carlet1998codes} to the unique known APN permutation over $\gf_{2^6}$ \cite{browning2010apn}. 
%
 It is clear that the Kim function is an example of $f(x)$ over $\gf_{2^6}$ with $(a_1, a_2, a_3)=(0, \frac{1}{u},  \frac{1}{u})$. \mn
Extending from the form of the Kim function, Carlet in \cite{carlet2014open} presented the following open problem. 
\begin{Prob}\label{Prob1}
	\cite{carlet2014open}
	Find more APN functions or, better, infinite classes of APN functions of the form $X^3+aX^{2+q}+bX^{2q+1}+cX^{3q}$ where $q=2^{n/2}$ with $n$ even, or more generally of the form  $X^{2^k+1}+aX^{2^k+q}+bX^{2^kq+1}+cX^{2^kq+q},$ where $\gcd(k,n)=1$. 
\end{Prob} 
The above problem is the main motivation of our study in this paper.
As a matter of fact, polynomials of the two general forms in Problem \ref{Prob1} have been studied from different perspectives. 
In Crypto'16,  Perrin et al. investigated the only APN permutation  by means of reverse-engineering and proposed the open butterfly and the closed butterfly structures \cite{perrin2016cryptanalysis}. The closed butterfly structure is represented as a bivariate function, of which the univariate form has the forms as in (\ref{f(x)}). Canteaut et al. later in \cite{Canteaut2019} 
proposed the generalized butterfly structure and showed that it can only produce APN functions over $\gf_{2^6}$.
For the functions from (generalized) butterfly structures, the coefficients of their univariate forms reside inherently in $\gf_{2^{m}}$.
Krasnayov{\'a} in \cite{krasnayova2016constructions}  considered $f(x)$ with the coefficients $a_i$’s in $\gf_{2^{m}}$ and study its permutation and APN properties. More specifically, Krasnayov{\'a} proposed a necessary and sufficient 
 condition that is expressed in terms of a trace function of the coefficients $a_i$’s and a new variable $T$, which is taken from an inexplicit subset of $\gf_{2^{m}}$ satisfying four equations.
Recently other cryptographic properties of $f(x)$  in (\ref{f(x)}) have been investigated as well.
A sufficient condition for $f(x)$ to be a permutation of $\gf_{2^{n}}$ with boomerang uniformity $4$, which is a new cryptanalysis property \cite{cid2018boomerang,li2019new,boura2018boomerang,mesnagerboomerang},
was given in \cite{tu2020class,TLZ2019} . The necessary condition for $f(x)$ to be a permutation has been characterized in \cite{li2020conjecture}. Furthermore, the permutation and boomerang uniformity property of a more general quadrinomial have been studied in \cite{li2019cryptographically,li20204,li2020permutation}. 

In this paper we will investigate the APN property of the quadrinomial $f(x)$ as in \eqref{f(x)}, which completely resolves the first part of Problem \ref{Prob1}.
The techniques, unfortunately, does not seem to work for the more general cases in its second part.
Before presenting the main theorem in this paper, we first make a basic transformation on the quadrinomial $f(x)$, which enables us
to restrict the coefficient $a_1$ to the finite field $\gf_{2^{m}}$ (instead of $\gf_{2^{n}}$). \mn Indeed, for any $a_1\in\gf_{2^{n}}$, one can take $b=a_1^{2^{m-1}}$ and then have
\begin{eqnarray*}
	f(b x) &=& (\overline{b}\overline{x})^3+a_1(\overline{b}\overline{x})^2b x + a_2b^2 x^2 \overline{b}\overline{x} + a_3(b x)^3 \\
	&=& \overline{b}^3\left( \overline{x}^3+a_1 (b/\overline{b}) \overline{x}^2x+a_2(b/\overline{b})^2x^2\overline{x} +a_3 (b/\overline{b})^3x^3    \right).
\end{eqnarray*} Since $f(x)$ and $f(bx)$ has the same differential uniformity and $a_1(b/\overline{b})=b^{2^m+1}$ belongs to $\gf_{2^{m}}$, it suffices to 
consider the coefficient $a_1$ in $\gf_{2^{m}}$ in subsequent discussions.
Let 
\begin{equation}
\label{theta}	
\left\{
\begin{array}{lr}
\theta_1 = 1+a_1^2+a_2\bar{a}_2+a_3\bar{a}_3  \\ 
\theta_2 = a_1+\bar{a}_2a_3 \\
\theta_3 = \bar{a}_2+a_1\bar{a}_3 \\
\theta_4 = a_1^2+a_2\bar{a}_2.
\end{array}
\right.
\end{equation} 
It is easy to check that 
\begin{equation}
\label{theta_equation}
\theta_2\bar{\theta}_2 + \theta_3\bar{\theta}_3 + \theta_1\theta_{4} + \theta_{4}^2 = 0. 
\end{equation}
Let $\tr_{m}(\cdot)$ denote the absolutely trace function over $\gf_{2^{m}}$, namely, $\tr_{m}(x)= x+x^2 + \cdots + x^{2^{m-1}}$. 
Our main theorem in this paper is as follows. 
\begin{Th}
	\label{theorem}
	Let $n=2m$ with $m\ge4$ and $f(x) = \bar{x}^3 + a_1\bar{x}^2x+a_2\bar{x}x^2+a_3x^3$, where $a_1\in\gf_{2^m}, a_2,a_3\in\gf_{2^n}$. 
	Let $\theta_i$'s be defined as in \eqref{theta} and 
	define
	\begin{equation}\label{Eq-Gamm1}
	\Gamma_1 = \left\{ (a_1,a_2,a_3) ~|~ \theta_1\neq0, ~ \tr_{m}\left(\frac{\theta_2\bar{\theta}_2}{\theta_1^2}\right) =0,~ \theta_1^2\theta_4 + \theta_1\theta_2\bar{\theta}_2 + \theta_2^2\theta_3 + \bar{\theta}_2^2\bar{\theta}_3 =0 \right\}
	\end{equation}
	and 
	\begin{equation}\label{Eq-Gamma2}
	\Gamma_2 = \left\{ (a_1,a_2,a_3) ~|~ \theta_1\neq0, ~\tr_{m}\left(\frac{\theta_2\bar{\theta}_2}{\theta_1^2}\right) =0, ~\theta_1^2\theta_3 + \theta_1\bar{\theta}_2^2 + \theta_2^2\theta_3 + \bar{\theta}_2^2\bar{\theta}_3 =0 \right\}.
	\end{equation}
	Then $f$ is APN over $\gf_{2^n}$ if and only if 
	\begin{enumerate}[(1)]
		\item $m$ is even, $(a_1,a_2,a_3)\in\Gamma_1\cup\Gamma_2$; or
		\item $m$ is odd, $(a_1,a_2,a_3)\in\Gamma_1$.
	\end{enumerate}
\end{Th}

\begin{Rem}
	{\emph{It is worth pointing out that the sets $\Gamma_1$ and $\Gamma_2$ are not empty since $(a_1,a_2,a_3)=(0,0,a_3)$ with $1+a_3\bar{a}_3\neq0$ belong to these sets clearly. In addition,   Tu et al. \cite{TLZ2019} and Li et al. \cite{li2020conjecture} showed that $f(x)$ defined as in (\ref{f(x)}) permutes $\gf_{2^{n}}$ if and only if $m$ is odd and $\left(a_1,a_2,a_3\right)\in\Gamma$, where $$\Gamma = \left\{ (a_1,a_2,a_3) ~|~ \theta_1\neq0, ~\tr_{m}\left(\frac{\theta_4}{\theta_1}\right) = 1,~ \theta_2^2 = \theta_1\bar{\theta}_3 \right\}$$ and $\theta_i$'s with $i=1,2,3,4$ are defined as in \eqref{theta}. 
			A natural question to ask is whether there exist APN permutations of the form $f(x)$
for $m\ge 4$. According to our result, the answer to this question is unfortunately negative. This is due to the fact  $\Gamma\cap \Gamma_1 = \emptyset$.  Indeed, if there exist some $(a_1,a_2,a_3)\in\Gamma\cap \Gamma_1$, one can plug $\theta_2^2=\theta_1\bar{\theta}_3$ into $\theta_1^2\theta_4 + \theta_1\theta_2\bar{\theta}_2 + \theta_2^2\theta_3 + \bar{\theta}_2^2\bar{\theta}_3 =0$. After a simplification,  it follows that $\theta_1\theta_4=\theta_2\bar{\theta}_2$ and thus $\tr_m\left( \frac{\theta_4}{\theta_1} \right) = \tr_{m}\left(\frac{\theta_2\bar{\theta}_2}{\theta_1^2}\right)$, which is a contradiction. Hence there is no APN permutation of the form \eqref{f(x)} for $m\geq 4$.}
		}
\end{Rem}

The rest of this paper is organized as follows: in Section \ref{Preliminaries},
we first introduce some notations and useful lemmas on algebraic curves, e.g., the Hasse-Weil bound, as well as give a proof sketch of Theorem \ref{theorem}  for readers' convenience.  
According to the proof sketch steps, Section \ref{proof} is devoted to the detailed proof of Theorem \ref{theorem}, of which some lengthy subcases are placed in the appendix section.
Section 4 concludes our work in this paper.

\section{Preliminaries}
\label{Preliminaries}

In this section, we first introduce some notations and several basic facts on finite fields. We always assume $n=2m$ and $f(x)$ is the quadrinomial over $\gf_{2^{n}}$ defined as in (\ref{f(x)}). For any finite set $E$, the nonzero elements of $E$ is denoted by $E^{*}$ and $\#E$ denotes the number of elements of $E$.

Let  $\overline{\gf}_{2^m}$ be the algebraic closure of $\gf_{2^{m}}$ and $\tr_{m}(\cdot)$ the absolutely trace function over $\gf_{2^{m}}$, namely, $\tr_{m}(x)= x+x^2 + \cdots + x^{2^{m-1}}$.
Suppose $k$ is an element in $\gf_{2^m}$ with $\tr_{m}(k)=1$  and $\omega$ is a solution of $x^2+x+k=0$ in $\gf_{2^n}$. Then $\omega$ is a primitive element of $\gf_{2^{n}}$ over $\gf_{2^{m}}$ satisfying $\bar{\omega}=\omega+1$ and $\omega\bar{\omega}=k$. The element $\omega$ induces a one-to-one correspondence between $\gf_{2^{n}}$ and $\gf_{2^{m}}^2$ by $x=x_1\omega+ x_2 \mapsto (x_1, x_2)$.
For convenience, given any element (except for $a$) in $\gf_{2^{n}}$, we will use the subscripts to denote its coordinates in $\gf_{2^{m}}$, e.g., we  use $x_1, x_2$ for $x$ and  $\theta_{11}, \theta_{12}$ for $\theta_1$. 
The unit circle of $\gf_{2^{n}}$ is defined by $\mu_{2^m+1} = \left\{ x\in\gf_{2^n} ~|~ x^{2^m+1} = 1 \right\}$. For any element $a\in\mu_{2^m+1}\backslash\{1\}$, it is well known that there exists a unique element $A\in\gf_{2^{m}}$ such that $a^2=\frac{A+\omega}{A+\bar{\omega}}$.

\mn


\subsection{Some results on algebraic curves}

In this subsection, we give some well known results on algebraic curves and algebraic function fields, mainly the Hasse-Weil bound, which plays an important role in our proof.  These classical results can be found in most of the textbooks on algebraic curves and algebraic function fields.

\begin{Lemma}
	\cite[Hasse-Weil bound]{hou2018lectures,stichtenoth2009algebraic}
	\label{Hasse-weil}
	Let $G(X,Y)$ be an absolutely irreducible polynomial in $\gf_{2^m}[X,Y]$ of degree $d$ and let $\# V_{\gf_{2^m}^2}(G)$ be the number of zeros of $G$. Then 
	$$ \left| \# V_{\gf_{2^m}^2}(G)-2^m \right| \le 2g2^{m/2}, $$
	where $g$ is the genus of the function field $\gf_{2^m}(\mathbb{X}, \mathbb{Y} )/ \gf_{2^m}$ and $ \mathbb{X}, \mathbb{Y} $ are transcendentals over $\gf_{2^m}$ with $G(\mathbb{X}, \mathbb{Y} ) = 0$. 	
\end{Lemma}

Let $F/K$ be a function field and $K$ be perfect. Let $g$ denote the genus of $F/K$. Then we have the following upper bound on  the genus.

\begin{Lemma}
	\cite{hou2018lectures,stichtenoth2009algebraic}
	\label{genus-bound}
	Let $ F = K(\mathbb{X}, \mathbb{Y} ) $, where  $ \mathbb{X}, \mathbb{Y} $ are transcendentals over $K$. Then the genus of the  function field $F/ K$ satisfies:
	$$ g\le ([F: K(\mathbb{X})] -1 )([F: K(\mathbb{Y})] -1). $$
\end{Lemma}

Given two plane curves $\mathcal{A}$ and $\mathcal{B}$ and a point $P$ on the plane, the \emph{intersection number} $I(P,\mathcal{A}\cap \mathcal{B})$ of $\mathcal{A}$ and $\mathcal{B}$ at the point $P$ is defined by seven axioms. We do not include its precise and long definitions here. For more details, we refer to \cite{hirschfeld2008algebraic}.
\begin{Lemma}
	\cite[B\'{e}zout's Theorem]{hirschfeld2008algebraic}
	\label{Btheorem}
	Let $\mathcal{A}$ and $\mathcal{B}$ be two projective plane curves over an algebraically closed field $K$, having no component in common. Let $A$ and $B$ be the polynomials associated with $\mathcal{A}$ and $\mathcal{B}$ respectively. Then 
	$$\sum_{P}I(P,\mathcal{A}\cap \mathcal{B})=(\deg A) (\deg B),$$
	where the sum runs over all points in the projective plane $\mathrm{PG}(2,K)$.
\end{Lemma}

\subsection{APN functions and the proof sketch of Theorem \ref{theorem}}

For a function $F(x)$ over $\gf_{2^n}$ and any $a\in\gf_{2^n}^{*}$, the function
$$D_aF(x) = F(x) + F(x+a)$$
is called the derivative of $F(x)$ in direction $a$. 
The \textit{differential uniformity}~\cite{nyberg1993differentially} of $F(x)$ is defined as
$$\max_{a\in\gf_{2^n}^{*}, b \in\gf_{2^n}}\# \{ x\in\gf_2^n ~|~ D_aF(x) = b \}.$$
Since $D_aF(x)=D_aF(x+a)$ for any $x, a$ in $\gf_2^n$, the minimum of differential uniformity of $F(x)$ is $2$. The functions with differential uniformity $2$ are called \emph{almost perfect nonlinear} (APN) functions.

For a quadratic function $F(x)$ with $F(0)=0$, it is well known that $F(x)$ is APN if and only if for any $a\in\gf_{2^{n}}^{*}$, the equation
 $$D_aF(x)+F(a)=0$$ has only two solutions $x=0,a$ in $\gf_{2^{n}}$, 
 equivalently, the equation 
 $$
 D_aF(ax)+F(a)=0
 $$  has only two solutions $x=0, 1$ in $\gf_{2^{n}}$.

Now we take a closer look at the quadrinomial $f(x)$ defined as in (\ref{f(x)}). For an element $b\in\gf_{2^m}$, it is clear that 
\begin{eqnarray*}
	f(b x) & = & (\bar{b}\bar{x})^3 + a_1(\bar{b}\bar{x})^2bx+a_2\bar{b}\bar{x}(b x)^2+a_3(b x)^3
=b^3 f(x).
\end{eqnarray*}
The above property of $f(x)$ enables us to restrict the element $a$ in the equation $ D_aF(ax)+F(a)=0$ to the unit circle $\mu_{2^m+1}$.
Indeed, 
for any $a\in\gf_{2^n}^{*}$, there exist a unique $b\in\gf_{2^m}^{*}$ and $c\in\mu_{2^m+1}$ such that $a=bc$, we have  
\begin{eqnarray*}
	D_af(ax) + f(a) &=& f(bcx+bc)+f(bcx)+f(bc) \\
	&=&b^3(f(cx+c)+f(cx)+f(c))=b^3(D_cf(cx)+f(c)).
\end{eqnarray*}
This means that the solution of $D_af(ax) + f(a) =0$ is independent of the choice of $b\in \gf_{2^{m}}^*$.
Therefore, we only need to determine the condition on $(a_1,a_2,a_3)$ such that $D_af(ax)+f(a)=0$ has only two solutions $x=0,1$ in $\gf_{2^n}$ for any $a\in\mu_{2^m+1}$. 
A simplification  of the equation $D_af(ax)+f(a)=0$ for $a$ in $\mu_{2^m+1}$ gives
\begin{equation}
\label{main_equation}
\epsilon_1 \bar{x}^2 + \epsilon_2 \bar{x} + \epsilon_3 x^2 + \epsilon_4 x = 0,
\end{equation}
where 
\begin{equation}	\label{epsilon}
\left\{
\begin{array}{lr}
\epsilon_1 = \bar{a}^3 + a_1\bar{a}^2a  \\ 
\epsilon_2 = \bar{a}^3 + a_2 \bar{a}a^2 \\
\epsilon_3 = a_2\bar{a}a^2 + a_3 a^3 \\
\epsilon_4 = a_1 \bar{a}^2 a + a_3 a^3.
\end{array}
\right.
\end{equation} \mn
Note that $\epsilon_1+\epsilon_2+\epsilon_3+\epsilon_4=0$.  Taking $2^m$-th power on both sides of (\ref{main_equation}) gives 
\begin{equation}
\label{main_equation_1}
\bar{\epsilon}_1 {x}^2 + \bar{\epsilon}_2 {x} + \bar{\epsilon}_3 \bar{x}^2 + \bar{\epsilon}_4 \bar{x} = 0.
\end{equation}
Computing the summation of the left side of (\ref{main_equation}) multiplied by $\bar{\epsilon}_3$ and the left one of (\ref{main_equation_1}) multiplied by $\epsilon_1$, we get 
\begin{equation}
\label{sum_equation}
\nu_1x^2+\nu_2\bar{x} + \nu_3 x =0,
\end{equation}  
where 
\begin{equation}	\label{nu}
\left\{
\begin{array}{lr}
\nu_1 = \epsilon_3\bar{\epsilon}_3+\bar{\epsilon}_1\epsilon_1 = a^4\bar{a}^2\theta_2 + a^3 \bar{a}^3 \theta_1 + a^2\bar{a}^4 \bar{\theta}_2\\
\nu_2 = \epsilon_2\bar{\epsilon}_3 + \bar{\epsilon}_4\epsilon_1 = a^3\bar{a}^3\theta_4 + a^2\bar{a}^4 \bar{\theta}_2 + a\bar{a}^5\theta_3\\
\nu_3 = \epsilon_4\bar{\epsilon}_3+\bar{\epsilon}_2\epsilon_1 = a^4\bar{a}^2\theta_2 + a^3\bar{a}^3 (\theta_1+\theta_4) + a\bar{a}^5\theta_3,
\end{array}
\right.
\end{equation} 
and $\theta_i$'s with $i=1,2,3,4$ are defined as in (\ref{theta}).

Li et. al  in \cite{li2020permutation} showed that if $\nu_1\neq0$, the number of solutions of (\ref{sum_equation}) is the same as that of (\ref{main_equation}). Moreover, the following lemma determines the number of solutions in $\gf_{2^{n}}$ of (\ref{sum_equation}), which for $\nu_1\neq0$ is actually $x^2+\tau \bar{x} +(1+\tau) x =0$ with $\tau=\frac{\nu_2}{\nu_1}$.  The proof of Lemma \ref{key_lemma} can be found in \cite{tu2020class}, where the authors assumed that $n=2m$ with $m$ odd. In fact, the condition $m$ odd can be deleted. Since the proof is very similar, we omit it here. \mn
\begin{Lemma}
	\cite{tu2020class}
	\label{key_lemma}
	Let $\tau\in\gf_{2^n}^{*}$. Then the equation $x^2+\tau \bar{x} +(1+\tau) x =0$ has two or four solutions in $\gf_{2^n}$. Moreover, the above equation has two solutions if and only if 
	\begin{enumerate}[(1)]
		\item $1+\tau+\bar{\tau}=0$; or 
		\item $1+\tau +\bar{\tau}\neq0$ and $\tr_{m}\left(\tau\bar{\tau}\right)=0$. 
	\end{enumerate}
\end{Lemma}

Before ending this section, we present a proof sketch of Theorem \ref{theorem} as the proof itself is technical and lengthy. The details of the proof will be given
 according to the sketch steps in the next section.

{\bfseries The proof sketch of Theorem \ref{theorem}.}
It suffices to show that (\ref{main_equation}) has only two solutions $x=0,1$ for any $a\in\mu_{2^m+1}$ if and only if $(a_1,a_2,a_3)\in \Gamma_1\cup\Gamma_2$ (resp. $\Gamma_1$) when $m$ is even (resp. odd). The proof can be divided into the following five steps.
\begin{enumerate}[Step 1).]
	\item Determine the necessary and sufficient condition $(a_1,a_2,a_3)$ for $a=1$ such that (\ref{main_equation}) has only two solutions $x=0,1$. See Proposition \ref{Prop_a=1}. 
	\item  For the case $a\in\mu_{2^m+1}\backslash\{1\}$, we first consider the elements $a$'s that lead to $\nu_1=0$, and find the necessary and sufficient condition such that (\ref{main_equation}) has only two solutions $x=0,1$ when $\nu_1=0$.  See (\ref{con_case1}).
	\item For the elements $a$'s such that $\nu_1\neq0$, it follows from Lemma \ref{key_lemma} that (\ref{main_equation}) has only two solutions $x=0,1$ if and only if $$1+\frac{\nu_2}{\nu_1}+\frac{\bar{\nu}_2}{\nu_1}=0 ~~\text{or}~~  1+\frac{\nu_2}{\nu_1}+\frac{\bar{\nu}_2}{\nu_1}\neq0 ~~\text{and}~~ \tr_{m}\left(\frac{\nu_2\bar{\nu}_2}{\nu_1^2}\right) = 0.$$ Thus for such $a$'s with $\nu_1\neq0$ and $1+\frac{\nu_2}{\nu_1}+\frac{\bar{\nu}_2}{\nu_1}\neq0$, $\tr_{m}\left(\frac{\nu_2\bar{\nu}_2}{\nu_1^2}\right) $ always equals zero. Recall that for any $a\in\mu_{2^m+1}\backslash\{1\}$, there exist a unique element $A\in\gf_{2^m}$ such that $a^2=\frac{A+\omega}{A+\bar{\omega}}$,
	where $\omega$ is a solution in $\gf_{2^n}$ of $x^2+x+k=0$ for an element $k\in\gf_{2^{m}}$ satisfying $\tr_{m}(k)=1$.  Substituting the expression of $a$ into $\tr_{m}\left(\frac{\nu_2\bar{\nu}_2}{\nu_1^2}\right)$, we  obtain 
	$\tr_{m}\left(\frac{L_1(A)}{L_2(A)^2}\right)=0$  for at least $(2^m-4)$ $A\in\gf_{2^{m}}$, where $L_1(Y),L_2(Y)\in\gf_{2^{m}}[Y]$. Moreover, the Hasse-Weil bound tells that when $m>5$, there exists some polynomial $L(Y)\in\gf_{2^{m}}[Y]$, such that $L(Y)L_2(Y)+L(Y)^2=L_1(Y)$. 
	\item Using the method of undetermined coefficients, determine the condition $(a_1,a_2,a_3)$, see Proposition \ref{four cases}, such that there indeed exists some polynomial $L(Y)\in\gf_{2^{m}}[Y]$ satisfying $L(Y)L_2(Y)+L(Y)^2=L_1(Y)$. There are four cases in this part of the proof. 
	\item Together with all conditions on $(a_1,a_2,a_3)$ obtained in Steps 1), 2), 4), we finally prove that (\ref{main_equation}) has only two solutions $x=0,1$ for any $a\in\mu_{2^m+1}$, i.e., $f(x)$ is APN if and only if $(a_1,a_2,a_3)\in \Gamma_1\cup\Gamma_2$ (resp. $\Gamma_1$) when $m$ is even (resp. odd). 
\end{enumerate}

\mn

\section{The detailed proof of Theorem \ref{theorem}}
\label{proof}
In this section, we give the whole proof of Theorem \ref{theorem}, mainly completing the five steps of the proof sketch in the above section. 

\subsection{The proof of Step 1).} 
In this subsection, we determine the condition on $(a_1,a_2,a_3)$ such that (\ref{main_equation}) has only two solutions $x=0,1$ in $\gf_{2^n}$ for $a=1$. 
In this case $\nu_1=\theta_1+\theta_2+\bar{\theta}_2$.
We divide the proof into two cases: $\nu_1=0$ or $\nu_1\neq0$. 

{\bfseries Case A.1: $\nu_1=0$.} In this case we have $\epsilon_1\bar{\epsilon}_1=\epsilon_3\bar{\epsilon}_3$.   We need to consider two subcases. (i) If $\nu_2=0$,  then (\ref{sum_equation}) becomes $0=0$ and thus we consider the equation (\ref{main_equation}) directly. For any $x\in\gf_{2^n}^{*}$, let $x=yz$, where $y\in\gf_{2^m}^{*}$ and $z\in\mu_{2^m+1}$, then (\ref{main_equation}) becomes 
\begin{equation}
\label{subcasei_eq1}
\left( \epsilon_1 z^{-2} + \epsilon_3 z^2 \right) y = \epsilon_2z^{-1} + \epsilon_4 z. 
\end{equation}
If $\epsilon_1=0$, then we have $\epsilon_3=0$ and $\epsilon_2=\epsilon_4$. Thus (\ref{main_equation}) becomes $\epsilon_2(\bar{x}+x)=0$, whose number of solutions is greater than $2$. If $\epsilon_1\neq0$, it is clear that there exist some $z\in\mu_{2^m+1}\backslash\{1\}$ such that $\epsilon_1 z^{-2} + \epsilon_3 z^2\neq0$. In addition, by checking directly, we can find that $\frac{\epsilon_2z^{-1} + \epsilon_4 z}{\epsilon_1 z^{-2} + \epsilon_3 z^2}\in\gf_{2^m}$ for any $z\in\mu_{2^m+1}$ and thus (\ref{main_equation}) has solutions $x=\frac{\epsilon_2+\epsilon_4z^2}{\epsilon_1 z^{-2} + \epsilon_3 z^2}$, where $z\in\mu_{2^m+1}$ satisfies $\epsilon_1 z^{-2} + \epsilon_3 z^2\neq0$. Hence, the number of solutions in $\gf_{2^n}$ of (\ref{main_equation}) is greater than $2$ in the subcase. (ii) If $\nu_2\neq0$, then from (\ref{sum_equation}),  we have $\bar{x}=x$. Plugging $\bar{x}=x$ into (\ref{main_equation}), we get $(\epsilon_1+\epsilon_3)(x^2+x)=0$, which has two solutions in $\gf_{2^n}$ if and only if $\epsilon_1+\epsilon_3\neq0$, i.e., $1+a_1+a_2+a_3\neq0$. 

Therefore, in this case, (\ref{main_equation}) has only two solutions in $\gf_{2^n}$ if and only if 
$$   \bar{\theta}_2+\theta_3+\theta_4\neq0 ~\text{and}~ 1+a_1+a_2+a_3\neq0,  $$
where $\theta_i$'s with $i=1,2,3,4$ are defined as in (\ref{theta}).

{\bfseries Case A.2: $\nu_1\neq0$.}
In \cite[Lemma 7]{li2020permutation}, the authors showed that (\ref{main_equation}) and (\ref{sum_equation}) have the same set of solutions in $\gf_{2^n}$. Thus it suffices to consider the equation (\ref{sum_equation}), which is equivalent to 
\begin{equation}
\label{case2_eq1}
x^2+ \frac{\nu_2}{\nu_1} \bar{x} +\left(1+\frac{\nu_2}{\nu_1}\right)
x=0. 
\end{equation}
By Lemma \ref{key_lemma}, (\ref{case2_eq1}) has only two solutions in $\gf_{2^n}$ if and only if $1+\frac{\nu_2}{\nu_1}+\frac{\bar{\nu}_2}{\nu_1}=0$ or  $1+\frac{\nu_2}{\nu_1}+\frac{\bar{\nu}_2}{\nu_1}\neq0$ and $\tr_{m}\left(\frac{\nu_2\bar{\nu}_2}{\nu_1^2}\right) = 0$. Firstly,  in this case we have
\begin{eqnarray*}
 1+\frac{\nu_2}{\nu_1}+\frac{\bar{\nu}_2}{\nu_1} 
= \frac{\theta_1+\theta_3 + \bar{\theta}_3}{\theta_1+\theta_2 + \bar{\theta}_2}
\end{eqnarray*}
and thus $1+\frac{\nu_2}{\nu_1}+\frac{\bar{\nu}_2}{\nu_1}=0$ if and only if $\theta_1+\theta_3 + \bar{\theta}_3=0$. Moreover, if $\theta_1+\theta_3 + \bar{\theta}_3\neq0$, 
$$\tr_{m}\left(\frac{\nu_2\bar{\nu}_2}{\nu_1^2}\right) =\tr_{m}\left(\frac{(\bar{\theta}_2+\theta_3+\theta_4)({\theta}_2+\bar{\theta}_3+\theta_4)}{(\theta_1+\theta_2 + \bar{\theta}_2)^2}\right) = 0.$$

From the above analysis, we have the following conclusion.

\begin{Prop}
	\label{Prop_a=1}
	Let $n=2m$, $a_1\in\gf_{2^m}$, $a_2,a_3\in\gf_{2^n}$, $\theta_i$'s with $i=1,2,3,4$ be defined as in (\ref{theta}) and $f(x) = \bar{x}^3 + a_1\bar{x}^2x+a_2\bar{x}x^2+a_3x^3$. Then the equation $f(x+1)+f(x)+f(1)=0$ has only two solutions $x=0,1$ in $\gf_{2^n}$ if and only if 
	\begin{enumerate}[(1)]
		\item  $\theta_1+\theta_2+\bar{\theta}_2=0, \bar{\theta}_2+\theta_3+\theta_4\neq0$ {and} $1+a_1+a_2+a_3\neq0$;
		\item $\theta_1+\theta_2+\bar{\theta}_2\neq0$ and $\theta_1+\theta_3+\bar{\theta}_3 =0$;
		\item $\theta_1+\theta_2+\bar{\theta}_2\neq0$, $\theta_1+\theta_3+\bar{\theta}_3 \neq0$ and
		$$ \tr_{m}\left(\frac{(\bar{\theta}_2+\theta_3+\theta_4)({\theta}_2+\bar{\theta}_3+\theta_4)}{(\theta_1+\theta_2 + \bar{\theta}_2)^2}\right) = 0.$$
	\end{enumerate}
\end{Prop}
\subsection{The proof of Step 2).}

Next, we determine the condition on $(a_1,a_2,a_3)$ such that the equation (\ref{main_equation}) has only two solutions $x=0,1$ in $\gf_{2^n}$ for any $a\in\mu_{2^m+1}\backslash\{1\}$. In this subsection, we mainly consider the case $\nu_1=0$. Recall that $k\in\gf_{2^m}$ satisfies $\tr_{m}(k)=1$  and $\omega$ is a solution of $x^2+x+ k=0$ in $\gf_{2^n}$.  Moreover, for any $a\in\mu_{2^m+1}\backslash\{1\}$, there exists a unique element $A\in\gf_{2^m}$ such that $a^2=\frac{A+\omega}{A+\bar{\omega}}$. 
 
 Furthermore, we have 
\begin{eqnarray*}
\nu_1 &=& a^2\theta_2 + \theta_1 +\bar{a}^2\bar{\theta}_2 \\
&=& \frac{A+\omega}{A+\omega+1}\theta_2 +\theta_1 + \frac{A+\omega+1}{A+\omega} \bar{\theta}_2\\
&=&\frac{  \varphi_1 A^2 + \theta_1 A + \varphi_2}{A^2+A+ k}
\end{eqnarray*}
and 
\begin{eqnarray*}
\nu_2 & = &  \theta_4+\bar{a}^2\bar{\theta}_2 + \bar{a}^4\theta_3 \\
&=& \theta_4+\frac{A+\omega+1}{A+\omega}\bar{\theta}_2 + \frac{A^2+\omega^2+1}{A^2+\omega^2}\theta_3 \\
&=& \frac{\varphi_3 A^2 + \bar{\theta}_2 A + \varphi_4}{A^2+ k+\omega},
\end{eqnarray*}
where
\begin{equation}
\label{varphi}	
\left\{
\begin{array}{lr}
\varphi_1 = \theta_1 + \theta_2 + \bar{\theta}_2 \\
\varphi_2 =  (\theta_1+\theta_2 + \bar{\theta}_2) k + \omega\theta_2+ (\omega+1)\bar{\theta}_2 \\
\varphi_3 = \theta_3 + \bar{\theta}_2 + \theta_4 \\
\varphi_4 = (\theta_3 + \bar{\theta}_2 + \theta_4) k + \omega \theta_4 + (\omega+1)\theta_3. 
\end{array}
\right.
\end{equation} 

For the case $\nu_1=0$,  similar with the case $a=1$, the equation (\ref{main_equation}) has only two solutions in $\gf_{2^n}$ if and only if 
$$
\nu_2\neq0 ~~\text{and}~~ \epsilon_1+\epsilon_3\neq0,
$$
i.e., 
\begin{equation}
\label{con_case1} \nu_2\neq0 ~~\text{and}~~ f(a)\neq0
\end{equation}
for such $a$'s satisfying $\nu_1=0$. 

In the final of this subcase, we consider the number of $A\in\gf_{2^{m}}$ such that $\nu_1=0$.  
\begin{Prop}
	\label{nu1=0}
	Let all notations be defined as in the above discussion and $Z(\nu_1)$ be the number of $A\in\gf_{2^{m}}$ such that $\nu_1=0$. Then
	\begin{enumerate}[(1)]
		\item if $\varphi_1 =0$, $\theta_1=0$, $\theta_2=0$, then $Z(\nu_1)=2^m$;
		\item if $\varphi_1 =0$, $\theta_1=0$, $\theta_2\neq0$, then $Z(\nu_1)=0$;
		\item if $\varphi_1 =0$ and $\theta_1\neq0$, or $\varphi_1\neq0$ and $\theta_1=0$, then $Z(\nu_1)=1$;
		\item if $\varphi_1 \neq0$,  $\theta_1\neq0$,  $ \tr_{m} \left(\frac{\theta_2\bar{\theta}_2}{\theta_1^2}\right) = 0$, $Z(\nu_1)=0$;
		\item if $\varphi_1 \neq0$,  $\theta_1\neq0$,  $ \tr_{m} \left(\frac{\theta_2\bar{\theta}_2}{\theta_1^2}\right) = 1$, $Z(\nu_1)=2$.
	\end{enumerate}
\end{Prop}
\begin{proof}
The proof is obvious and it suffices to show that  $\tr_{m}\left( \frac{\varphi_1\varphi_2}{\theta_1^2} \right) = \tr_{m} \left(\frac{\theta_2\bar{\theta}_2}{\theta_1^2}\right)+1$ if $\varphi_1\neq0$ and $\theta_1\neq0$, which holds since 
\begin{eqnarray*}
\tr_{m}\left( \frac{\varphi_1\varphi_2}{\theta_1^2} \right)&= & \tr_{m}\left( \frac{(\theta_1+\theta_2+\bar{\theta}_2)((\theta_1+\theta_2+\bar{\theta}_2)k + \omega \theta_2 + \omega\bar{\theta}_2+\bar{\theta}_2)}{\theta_1^2} \right) \\
&=& \tr_{m}\left(k\right)+\tr_{m}\left(\frac{\theta_2\bar{\theta}_2}{\theta_1^2}\right) + \tr_{m}\left( \frac{(\theta_2^2 + \bar{\theta}_2^2)k+\omega\theta_1\theta_2+(\omega+1)\theta_1\bar{\theta}_2 + \omega\theta_2^2 + (\omega+1)\bar{\theta}_2^2}{\theta_1^2} \right)\\
&=& 1 +\tr_{m}\left(\frac{\theta_2\bar{\theta}_2}{\theta_1^2}\right) + \tr_{m}\left( \frac{(\omega\theta_2+(\omega+1)\bar{\theta}_2)^2+\theta_1(\omega\theta_2+(\omega+1)\bar{\theta}_2)}{\theta_1^2}  \right) \\
&=&\tr_{m}\left(\frac{\theta_2\bar{\theta}_2}{\theta_1^2}\right) +1.
\end{eqnarray*}
\end{proof}

%

\subsection{The proof of Step 3).}
For the case $\nu_1\neq0$,  also similar with the case $a=1$, (\ref{main_equation}) has only two solutions in $\gf_{2^n}$ if and only if $$1+\frac{\nu_2}{\nu_1}+\frac{\bar{\nu}_2}{\nu_1}=0 ~~\text{or}~~  1+\frac{\nu_2}{\nu_1}+\frac{\bar{\nu}_2}{\nu_1}\neq0 ~~\text{and}~~ \tr_{m}\left(\frac{\nu_2\bar{\nu}_2}{\nu_1^2}\right) = 0.$$ 
By computing directly and simplifying,  $1+\frac{\nu_2}{\nu_1}+\frac{\bar{\nu}_2}{\nu_1}=0$ if and only if
\begin{equation}
\label{A_equation}
(\theta_1+\theta_3+\bar{\theta}_3)A^4 + \theta_1A^2 + (\theta_1+\theta_3+\bar{\theta}_3)k^2 + (\theta_3+\bar{\theta}_3) k + \omega\theta_3+\theta_3+ \omega \bar{\theta}_3 =0.
\end{equation}
Then, $\tr_{m}\left(\frac{\nu_2\bar{\nu}_2}{\nu_1^2}\right) = 0$ always holds for  all $A \in \gf_{2^m}$ but several $A$'s satisfying $\nu_1=0$, i.e., $\varphi_1 A^2 + \theta_1 A + \varphi_2 =0$ or (\ref{A_equation}) holds, whose number is at most $4$. Namely, the number of $A$'s such that $\tr_{m}\left(\frac{\nu_2\bar{\nu}_2}{\nu_1^2}\right) = 0$ holds is at least $(2^m-4)$.

 Moreover, by simplifying, 
\begin{equation}
 \tr_{m}\left(\frac{\nu_2\bar{\nu}_2}{\nu_1^2}\right) = \tr_{m}\left(\frac{L_1(A)}{L_2(A)^2}\right),
\end{equation} 
where $L_1(A) = l_{11}A^4+l_{12}A^3 + l_{13} A^2 + l_{14}A + l_{15}$,
	\begin{equation}
\label{L_1}	
\left\{
\begin{array}{lr}
l_{11} = \varphi_3\bar{\varphi}_3  \\ 
l_{12} = \theta_2\varphi_3 + \bar{\theta}_2\bar{\varphi}_3 \\
l_{13} = \theta_2\bar{\theta}_2 + \varphi_3\bar{\varphi}_4 + \bar{\varphi}_3\varphi_4 \\
l_{14} = \theta_2\varphi_4 + \bar{\theta}_2\bar{\varphi}_4 \\
l_{15} = \varphi_4\bar{\varphi}_4,
\end{array}
\right.
\end{equation} 
 $L_2 (A) =  \varphi_1 A^2 + \theta_1A+ \varphi_2$ and $\varphi_i$'s with $i=1,2,3,4$ are defined as in (\ref{varphi}). 

Let $L_1(Y) = l_{11}Y^4+l_{12}Y^3 + l_{13} Y^2 + l_{14}Y + l_{15}\in\gf_{2^m}[Y] $, $L_2(Y) = \varphi_1 A^2 + \theta_1A+ \varphi_2 \in\gf_{2^m}[Y],$
$$G(X,Y) = L_2(Y)^2(X^2+X)+L_1(Y)\in\gf_{2^m}[X,Y]$$
and 
$$V_m(G) = \{ (x,y)\in\gf_{2^m}^2  ~|~ G(x,y) = 0 \}.$$
Then $\deg G =6$. Since the number of $A$'s such that $\tr_{m}\left(\frac{L_1(A)}{L_2(A)^2} \right) = 0$ holds is at least $(2^m-4)$,   
\begin{equation}
\label{V_n}
\#V_m(G)\ge 2 (2^m-4).
\end{equation}

	If $G(X,Y)$ is  irreducible over $\overline{\gf}_{2^m}$, let $ \mathbb{X}, \mathbb{Y} $ be transcendentals over $\gf_{2^m}$ with $G(\mathbb{X}, \mathbb{Y} ) = 0$. Then by Lemma \ref{genus-bound}, the function field $\gf_{2^m}(\mathbb{X}, \mathbb{Y} )/ \gf_{2^m}$ has genus
$$  g\le ([\gf_{2^m}(\mathbb{X}, \mathbb{Y} ): \gf_{2^m}(\mathbb{X})] -1 )([\gf_{2^m}(\mathbb{X}, \mathbb{Y} ): \gf_{2^m}(\mathbb{Y})] -1) \le (4-1)(2-1)=3. $$
Then by the Hasse-Weil bound, i.e., Lemma \ref{Hasse-weil}, we have 
$$  \# V_m(G) \le 2^m+1+2g2^{m/2}\le 2^m+1+ 6\cdot2^{m/2}<  2\left( 2^m-4 \right), $$
when $ m > 5$, which is contradictory with (\ref{V_n}). 

Therefore, $G(X,Y)$ is not irreducible over $\overline{\gf}_{2^m}$ and we assume that $G(X,Y) =  sG_1(X,Y)G_2(X,Y),$ where $s\in\gf_{2^m}^{*}$, $G_1, G_2 \in \overline{\gf}_{2^m}[X,Y]$ are irreducible and $\deg_X G_1 = \deg_X G_2 = 1.$ If $G_1\notin \gf_{2^m}[X,Y]$, choose $\sigma\in\mathrm{Aut} (\overline{\gf}_{2^m}/\gf_{2^m})$ such that $\sigma(G_1)\neq G_1$. Then $\sigma(G_1)=G_2$. Assume that $(x,y)\in V_{m}(G)$. Then $(x,y)\in V_{m}(G_1)$ or $V_{m}(G_2)$, say $(x,y)\in V_{m}(G_1).$ Then $(x,y)=(\sigma(x),\sigma(y))\in V_{m}(\sigma (G_1))$. Hence $(x,y)\in V_{m}(G_1)\cap V_{m}(\sigma (G_1))$ and we have 
$$ V_m(G) \subset V_m(G_1)\cap V_m(\sigma(G_1)). $$
Thanks to B\'{e}zout's Theorem, i.e., Lemma \ref{Btheorem},
$$ \#V_m(G)\le \# \left( V_m(G_1) \cap V_m(\sigma(G_1))  \right) \le (\deg G_1)^2 \le 9,$$
which is also contradictory with (\ref{V_n}). Thus $G_1, G_2\in\gf_{2^m}[X,Y]$. Namely, there exists some $L(Y)\in\gf_{2^m}[Y]$ such that 
$$X^2+X+\frac{L_1(Y)}{L_2(Y)^2} = \left(X+ \frac{L(Y)}{L_2(Y)} \right)   \left( X+ 1 + \frac{L(Y)}{L_2(Y)} \right). $$
Hence, 
\begin{equation}
\label{LY}
L(Y)L_2(Y)+L(Y)^2 = L_1(Y).
\end{equation}

\subsection{The proof of Step 4).}

Clearly, from (\ref{LY}), $\deg L = 2$ and we assume that $L(Y) = l_1 Y^2+l_2Y+l_3\in\gf_{2^m}[Y]$. After comparing the coefficients of two sides of (\ref{LY}) and together with (\ref{L_1}), we have 
	\begin{subequations} 
		\label{D}
	\renewcommand\theequation{\theparentequation.\arabic{equation}}     
	\begin{empheq}[left={\empheqlbrace\,}]{align}
	&~ D_4:  l_1^2 + \varphi_1 l_1 = \varphi_3\bar{\varphi}_3 \label{D4} \\
	&~ D_3:  \theta_1 l_1 + \varphi_1 l_2 = \theta_2\varphi_3 + \bar{\theta}_2\bar{\varphi}_3 \label{D3}\\
	&~ D_2:  \varphi_2 l_1+ \theta_1 l_2 +  \varphi_1 l_3+l_2^2 = \theta_2\bar{\theta}_2 + \varphi_3\bar{\varphi}_4 + \bar{\varphi}_3\varphi_4 \label{D2} \\
	&~ D_{1}:     \varphi_2 l_2+\theta_1 l_3 = \theta_2\varphi_4 + \bar{\theta}_2\bar{\varphi}_4 \label{D1}  \\
	&~ D_{0}:  l_3^2 + \varphi_2 l_3 = \varphi_4\bar{\varphi}_4. \label{D0} 
	\end{empheq}
\end{subequations}
In the above equation system, $D_i$'s  denote that the equation in the same row  is from comparing the coefficient of degree $i$, where $i=0,1,2,3,4$. Thus for $\nu_1\neq0$, (\ref{main_equation}) has only two solutions in $\gf_{2^{n}}$ if and only if (\ref{D}) has solutions for $l_i\in\gf_{2^{m}}$ with $i=1,2,3$. Next, we divide this part of the proof into four cases: (D.1) $\varphi_1\varphi_2\neq0$; (D.2) $\varphi_1\neq0$ and $\varphi_2=0$; (D.3) $\varphi_1=0$ and $\varphi_2\neq 0$; (D.4) $\varphi_1=0$ and $\varphi_2=0$. For convenience, we summary the  four cases firstly as follows.

\begin{Prop}
	\label{four cases}
\begin{enumerate}[(D.1)]
	\item If $\varphi_1\varphi_2\neq0$, (\ref{D}) has solutions if and only if 
	$ \theta_1\neq0, (\theta_1^2\theta_3+ \theta_1 \bar{\theta}_2^2 + \theta_2^2\theta_3+\bar{\theta}_2^2\bar{\theta}_3)(\theta_1^2\theta_4 + \theta_1\theta_2\bar{\theta}_2 + \theta_2^2\theta_3 + \bar{\theta}_2^2\bar{\theta}_3)=0, \tr_{m}\left( \frac{\varphi_3\bar{\varphi}_3}{\varphi_1^2} \right) = 0. $
	\item If $\varphi_1\neq0$ and $\varphi_2=0$, then (\ref{D}) has solutions if and only if 
	$m$ is odd, $(\theta_{21}+\theta_{31})k^3+\theta_{22}k^2 + (\theta_{4}+\theta_{21})k + \theta_{21}+\theta_{22} =0$,  $(\theta_{21} + \theta_{31} ) k^2 + (\theta_{21}  + \theta_{22}  +  \theta_{31}) k + \theta_{32} = 0$ and $ (\theta_{21}^2  + \theta_{21} \theta_{31}) k + \theta_{21} \theta_{31} + \theta_{22} \theta_{31} =0.$
	\item If $\varphi_1=0$ and $\varphi_2\neq 0$, then (\ref{D}) does not have solutions.
	\item If $\varphi_1=0$ and $\varphi_2=0$, then (\ref{D}) has solutions if and only if $m$ is odd and $\theta_{4}=\theta_{32}+\theta_{21}$. 
\end{enumerate}	
\end{Prop}

Next, we give the discussions of all four cases. 

{\bfseries Case D.1} If $\varphi_1\varphi_2\neq0$, then $l_1$ and $l_3$ are solutions of (\ref{D4}) and (\ref{D0}), respectively and thus $$ \tr_m\left( \frac{\varphi_3\bar{\varphi}_3}{\varphi_1^2}\right) = 0 ~\text{and}~ \tr_m\left( \frac{ \varphi_4\bar{\varphi}_4}{\varphi_2^2} \right) =0. $$  
We assume that $\frac{\varphi_3\bar{\varphi}_3}{\varphi_1^2}=\xi^2+\xi$ and $\frac{ \varphi_4\bar{\varphi}_4}{\varphi_2^2} = \eta^2+\eta, $ where $\xi,\eta\in\gf_{2^m}$. Then W.L.O.G., we have $l_1=\varphi_1 \xi$ and $l_3 = \varphi_2 \eta$. Moreover, from (\ref{D3}) and (\ref{D1}), we have 
$$l_2 = \frac{\theta_2\varphi_3+\bar{\theta}_2\bar{\varphi}_3}{\varphi_1} + \theta_1\xi $$
and $$l_2 = \frac{\theta_2\varphi_4+\bar{\theta}_2\bar{\varphi}_4}{\varphi_2} + \theta_1\eta,$$
respectively.

 If $\theta_1\neq0$,  together with the above two expressions of $l_2$ and vanishing $l_2$, we obtain 
\begin{equation}
\label{xieta}
\xi+\eta = \frac{\theta_2\varphi_2\varphi_3 + \bar{\theta}_2\varphi_2\bar{\varphi}_3+\theta_2\varphi_1\varphi_4 + \bar{\theta}_2\varphi_1\bar{\varphi}_4}{\theta_1\varphi_1\varphi_2}
\end{equation}
Moreover, it is clear that 
\begin{equation}
\label{xieta2} \xi^2+\xi+\eta^2+\eta = \frac{\varphi_3\bar{\varphi}_3}{\varphi_1^2} + \frac{ \varphi_4\bar{\varphi}_4}{\varphi_2^2} =\frac{\varphi_2^2\varphi_3\bar{\varphi}_3+\varphi_1^2\varphi_4\bar{\varphi}_4}{\varphi_1^2\varphi_2^2}.
\end{equation}
Together with (\ref{xieta}) and (\ref{xieta2}), we get
\begin{equation}
\label{con1}
\theta_1^2\theta_4 + \theta_1\theta_2\bar{\theta}_2 + \theta_2^2\theta_3 + \bar{\theta}_2^2\bar{\theta}_3 = 0
\end{equation}
or 
\begin{equation}
\label{con2}
\left( \theta_1^2\theta_3+\theta_1^2\bar{\theta}_3 + \theta_1\theta_2^2 + \theta_1 \bar{\theta}_2^2 \right)\omega^2 = \theta_1^2\theta_3+ \theta_1 \bar{\theta}_2^2 + \theta_2^2\theta_3+\bar{\theta}_2^2\bar{\theta}_3.
\end{equation}

Furthermore, plugging  $l_1=\varphi_1 \xi$, $l_2 = \frac{\theta_2\varphi_3+\bar{\theta}_2\bar{\varphi}_3}{\varphi_1} + \theta_1\xi$ and $l_3 = \varphi_2 \eta$ into (\ref{D2}) and simplifying, we obtain 
\begin{equation}
\theta_1\theta_3+\theta_1\bar{\theta}_3 + \theta_2^2 + \bar{\theta}_2^2=0
\end{equation}
or (\ref{con1}) holds. 

In addition, it is clear that when
$\varphi_1\varphi_2\neq0$, $\theta_1\neq0$  and $ (\theta_1^2\theta_3+ \theta_1 \bar{\theta}_2^2 + \theta_2^2\theta_3+\bar{\theta}_2^2\bar{\theta}_3)(\theta_1^2\theta_4 + \theta_1\theta_2\bar{\theta}_2 + \theta_2^2\theta_3 + \bar{\theta}_2^2\bar{\theta}_3)=0  $, $\frac{\varphi_3\bar{\varphi}_3}{\varphi_1^2} + \frac{ \varphi_4\bar{\varphi}_4}{\varphi_2^2} = \xi^2+\xi+\eta^2+\eta $ and thus
$$\tr_{m}\left(\frac{\varphi_3\bar{\varphi}_3}{\varphi_1^2}\right) = \tr_{m}\left( \frac{\varphi_4\bar{\varphi}_4}{\varphi_2^2} \right).$$

Therefore, if $\varphi_1\varphi_2\neq0$ and $\theta_1\neq0$, (\ref{D}) has solutions if and only if $$ (\theta_1^2\theta_3+ \theta_1 \bar{\theta}_2^2 + \theta_2^2\theta_3+\bar{\theta}_2^2\bar{\theta}_3)(\theta_1^2\theta_4 + \theta_1\theta_2\bar{\theta}_2 + \theta_2^2\theta_3 + \bar{\theta}_2^2\bar{\theta}_3)=0~~ \text{and}~~ \tr_{m}\left(\frac{\varphi_3\bar{\varphi}_3}{\varphi_1^2}\right) = 0. $$

If $\theta_1=0$, from (\ref{D3}) and (\ref{D1}), we have $l_2=\frac{\theta_2\varphi_3+\bar{\theta}_2\bar{\varphi}_3}{\varphi_1}$ and $l_2 = \frac{\theta_2\varphi_4+\bar{\theta}_2\bar{\varphi}_4}{\varphi_2}$, respectively. Thus we have 
$$\frac{\theta_2\varphi_3+\bar{\theta}_2\bar{\varphi}_3}{\varphi_1} = \frac{\theta_2\varphi_4+\bar{\theta}_2\bar{\varphi}_4}{\varphi_2},$$
i.e.,
\begin{equation}
\label{con3}
\theta_2^2\theta_3 + \bar{\theta}_2^2\bar{\theta}_3 = 0.
\end{equation}
The next lemma shows that if $\theta_1=0$ and (\ref{con3}) holds, $\tr_{m}\left(\frac{\varphi_3\bar{\varphi}_3}{\varphi_1^2}\right)=1$ and thus (\ref{D}) does not have solutions.
\begin{Lemma}
	If $\theta_1+\theta_2+\bar{\theta}_2\neq0$, $\theta_1=0$ and $\theta_2^2\theta_3 + \bar{\theta}_2^2\bar{\theta}_3 = 0.$ Then 
	$$T 	\triangleq \tr_{m}\left(\frac{\varphi_3\bar{\varphi}_3}{\varphi_1^2}\right) =1. $$
\end{Lemma}
\begin{proof}
	Clearly, $\theta_2\neq0$ and $\bar{\theta}_3=\frac{\theta_2^2\theta_3}{\bar{\theta}_2^2}$, which will be used sometimes in the following. From (\ref{theta_equation}), $\theta_2\bar{\theta}_2+\theta_3\bar{\theta}_3=\theta_4^2$. By computing directly, we have
	\begin{eqnarray*}
T 	&=&  \tr_{m}\left(\frac{(\bar{\theta}_2+\theta_3+\theta_4)({\theta}_2+\bar{\theta}_3+\theta_4)}{(\theta_1+\theta_2 + \bar{\theta}_2)^2}\right) \\
	&=& \tr_{m}\left(\frac{(\theta_2\bar{\theta}_2+\theta_3\bar{\theta}_3+\theta_4^2) + \theta_4(\theta_2+\bar{\theta}_2) + \theta_2\theta_3 + \bar{\theta}_2\bar{\theta}_3 +\theta_3\theta_4 + \bar{\theta}_3\theta_4}{(\theta_2 + \bar{\theta}_2)^2}\right)\\
	&=& \tr_{m}\left( \frac{\theta_4^2 + \theta_2\theta_3 + \bar{\theta}_2\bar{\theta}_3 +\theta_3\theta_4 + \bar{\theta}_3\theta_4}{(\theta_2 + \bar{\theta}_2)^2} \right),
	\end{eqnarray*}
where the last equality holds thanks to $\tr_{m}\left( \frac{\theta_4(\theta_2+\bar{\theta}_2)}{(\theta_2 + \bar{\theta}_2)^2} \right) = \tr_{m}\left( \frac{\theta_4^2}{(\theta_2 + \bar{\theta}_2)^2} \right)$ and $\theta_2\bar{\theta}_2+\theta_3\bar{\theta}_3=\theta_4^2$. Plugging $\bar{\theta}_3=\frac{\theta_2^2\theta_3}{\bar{\theta}_2^2}$ into the expression of $T$ and continuing simplifying, we get
\begin{eqnarray*}
T &=& \tr_{m}\left(  \frac{ \theta_2^2\bar{\theta}_2\theta_3 +\theta_2^2\theta_3\theta_4 +\bar{\theta}_2^2\theta_4^2 + \theta_2\bar{\theta}_2^2\theta_3 + \bar{\theta}_2^2\theta_3\theta_4 }{ \bar{\theta}_2^2 (\theta_2 + \bar{\theta}_2)^2}  \right) \\
&=&\tr_{m}\left( \frac{ \theta_2^2\bar{\theta}_2\theta_3 +\theta_2^2\theta_3\theta_4 +\theta_2\bar{\theta}_2^3 +\bar{\theta}_2^2\theta_3\bar{\theta}_3 + \theta_2\bar{\theta}_2^2\theta_3 + \bar{\theta}_2^2\theta_3\theta_4 }{ \bar{\theta}_2^2 (\theta_2 + \bar{\theta}_2)^2} \right)\\
&=&\tr_{m}\left(\frac{\theta_2\bar{\theta}_2}{ (\theta_2 + \bar{\theta}_2)^2}\right) + \tr_{m}\left( \frac{\theta_2\theta_3\bar{\theta}_2(\theta_2+\bar{\theta}_2) + \theta_2^2\theta_3^2}{\bar{\theta}_2^2 (\theta_2 + \bar{\theta}_2)^2}\right) + \tr_{m}\left(\frac{\theta_3\theta_4}{\bar{\theta}_2^2}\right).
\end{eqnarray*}
Moreover, since $\frac{\theta_2}{\theta_2 + \bar{\theta}_2}\in\gf_{2^n}\backslash\gf_{2^m}$ and $\frac{\theta_2\theta_3}{ \bar{\theta}_2(\theta_2 + \bar{\theta}_2)}\in\gf_{2^m}$, which can be proved by computing directly, we have $$\tr_{m}\left( \frac{\theta_2\bar{\theta}_2}{ (\theta_2 + \bar{\theta}_2)^2} \right) = \tr_{m}\left( \frac{\theta_2}{\theta_2 + \bar{\theta}_2}+\frac{\theta_2^2}{(\theta_2 + \bar{\theta}_2)^2}\right)=1$$
and 
$$\tr_{m}\left( \frac{\theta_2\theta_3\bar{\theta}_2(\theta_2+\bar{\theta}_2) + \theta_2^2\theta_3^2}{\bar{\theta}_2^2 (\theta_2 + \bar{\theta}_2)^2}\right) = \tr_{m}\left(\frac{\theta_2\theta_3}{ \bar{\theta}_2(\theta_2 + \bar{\theta}_2)} + \frac{\theta_2^2\theta_3^2}{ \bar{\theta}_2^2(\theta_2 + \bar{\theta}_2)^2}  \right) = 0.$$
Furthermore, 
$$\tr_{m}\left(\frac{\theta_3\theta_4}{\bar{\theta}_2^2}  \right) =\tr_{m}\left( \frac{\theta_3\bar{\theta}_3\theta_4^2}{\theta_2^2\bar{\theta}_2^2} \right) =\tr_{m} \left( \frac{\theta_3\bar{\theta}_3(\theta_2\bar{\theta}_2+\theta_3\bar{\theta}_3)}{\theta_2^2\bar{\theta}_2^2} \right) = 0.$$
	Thus $T = 1+0+0 = 1$. The proof is complete.
\end{proof}

All in all, in this case $\varphi_1\varphi_2\neq0$, (\ref{D}) has solutions if and only if 
	$$ \theta_1\neq0, ~(\theta_1^2\theta_3+ \theta_1 \bar{\theta}_2^2 + \theta_2^2\theta_3+\bar{\theta}_2^2\bar{\theta}_3)(\theta_1^2\theta_4 + \theta_1\theta_2\bar{\theta}_2 + \theta_2^2\theta_3 + \bar{\theta}_2^2\bar{\theta}_3)=0,~ \tr_{m}\left( \frac{\varphi_3\bar{\varphi}_3}{\varphi_1^2} \right) = 0. $$

{\bfseries Case D.2} If $\varphi_1\neq0$ and $\varphi_2=0$, then 
\begin{equation}
\label{p2=0}
(\theta_1+\theta_2 + \bar{\theta}_2) k + \omega\theta_2+ (\omega+1)\bar{\theta}_2=0.
\end{equation} We assume that $\theta_2 = \theta_{21}\omega + \theta_{22}$, $\theta_3 = \theta_{31}\omega + \theta_{32}$, where $\theta_{21},\theta_{22},\theta_{31},\theta_{32}\in\gf_{2^m}$. Then it is easy to obtain that $\bar{\theta}_2 = \theta_{21}\omega + \theta_{21} + \theta_{22} $, $ \bar{\theta}_3 = \theta_{31}\omega + \theta_{31} + \theta_{32}$,
$$ \theta_2 + \bar{\theta}_2 = \theta_{21}, ~~~~~~ \theta_3 + \bar{\theta}_3 = \theta_{31}, $$
and 
$$ \theta_2\bar{\theta}_2 = \theta_{21}^2k + \theta_{21}\theta_{22} + \theta_{22}^2, ~~~~~  \theta_3\bar{\theta}_3 = \theta_{31}^2k + \theta_{31}\theta_{32} + \theta_{32}^2.$$

Plugging $\theta_2 = \theta_{21}\omega + \theta_{22}$ into (\ref{p2=0}), we get 
\begin{equation}
\label{v2=0} (\theta_1+ \theta_{21}) k +  \theta_{21} + \theta_{22} = 0
\end{equation}
and thus  $k=\frac{\theta_{21} + \theta_{22} }{\theta_1+\theta_{21}}$ since  $\varphi_1= \theta_1+\theta_{21}\neq0$. Also plugging $\theta_2 = \theta_{21}\omega + \theta_{22}$ and  $\theta_3 = \theta_{31}\omega + \theta_{32}$ into (\ref{varphi}) and simplifying, we can obtain
\begin{equation}
\label{v134}\varphi_1 = \theta_1 + \theta_{21}, ~~ \varphi_3 = \varphi_{31} \omega + \varphi_{32}, ~~ \text{and} ~~  \varphi_4 = \varphi_{41} \omega + \varphi_{42},
\end{equation} 
where 
\begin{equation}
\label{p34}	
\left\{
\begin{array}{lr}
\varphi_{31} = \theta_{21} + \theta_{31} \\
\varphi_{32} = \theta_{21} + \theta_{22} + \theta_{32} + \theta_4 \\
\varphi_{41} = k\theta_{21} + k\theta_{31} + \theta_{32} + \theta_4 \\
\varphi_{42} = k \theta_{21} + k\theta_{22} + k\theta_{31} + (k+1) \theta_{32} + k\theta_4. 
\end{array}
\right.
\end{equation}

Furthermore, from (\ref{D0}) and (\ref{D1}), we have 
$$ l_3^2 = \varphi_4\bar{\varphi}_4= \varphi_{41}^2k+\varphi_{41}\varphi_{42} + \varphi_{42}^2$$ and $$\theta_1l_3= \theta_2\varphi_4 + \bar{\theta}_2\bar{\varphi}_4 = \theta_{21}(\varphi_{41}+\varphi_{42}) + \theta_{22}\varphi_{41},$$ respectively. By vanishing $l_3$, we obtain 
\begin{equation*}
\theta_1^2\left( \varphi_{41}^2k+ \varphi_{41}\varphi_{42} + \varphi_{42}^2 \right) = \theta_{21}^2 \left( \varphi_{41}^2 + \varphi_{42}^2 \right) + \theta_{22}^2\varphi_{41}^2.
\end{equation*}
After simplifying the above equation, we can get  $S_1=0$, where \begin{eqnarray}
	S_1&= &  \theta_1^2 \theta_{21}^2 k^3 + \theta_1^2 \theta_{21} \theta_{22} k^2 + \theta_1^2 \theta_{21} \theta_{32} k^2 + \theta_1^2 \theta_{21} \theta_4 k^2 + \theta_1^2 \theta_{21} \theta_4 k + \theta_1^2 \theta_{22}^2 k^2 + \theta_1^2 \theta_{22} \theta_{31} k^2 \notag \\
	&& + \theta_1^2 \theta_{22} \theta_{32} k + \theta_1^2 \theta_{22} \theta_4 k +
	\theta_1^2 \theta_{31}^2 k^3 + \theta_1^2 \theta_{31} \theta_{32} k^2 + \theta_1^2 \theta_{31} \theta_4 k^2 + \theta_1^2 \theta_{31} \theta_4 k + \theta_1^2 \theta_{32}^2 k^2 \notag\\
	&&+ \theta_1^2 \theta_{32} \theta_4 + \theta_1^2 \theta_4^2 k^2 + \theta_{21}^2 \theta_{32}^2 k^2 + \theta_{21}^2 \theta_4^2 k^2 +
	\theta_{21}^2 \theta_4^2 + \theta_{22}^2 \theta_{31}^2 k^2 + \theta_{22}^2 \theta_{32}^2 + \theta_{22}^2 \theta_4^2. \label{S1}
\end{eqnarray}
Moreover, since $\varphi_1\neq0$, from (\ref{D4}), we have $\tr_{m}\left(\frac{\varphi_3\bar{\varphi}_3}{\varphi_1^2}\right) = 0$. We assume that $\frac{\varphi_3\bar{\varphi}_3}{\varphi_1^2} =\xi^2+\xi $ and thus $l_1 = \varphi_1\xi$. In addition, from (\ref{D3}) and simplifying,
\begin{equation}
\label{case2.2_l2}l_2=\frac{\theta_2\varphi_3 + \bar{\theta}_2 \bar{\varphi}_3}{\varphi_1}+\theta_1\xi = \frac{\theta_{21}\varphi_{31} + \theta_{21} \varphi_{32} + \theta_{22}\varphi_{31} }{\varphi_1} + \theta_1\xi.
\end{equation} 
Furthermore, we have 
\begin{eqnarray*}
\varphi_3\bar{\varphi}_4&=& \left( \varphi_{31} \omega + \varphi_{32} \right) \left( \varphi_{41}\omega + \varphi_{41} + \varphi_{42} \right) \\
&=& \left( \varphi_{31}\varphi_{42} + \varphi_{32} \varphi_{41} \right) \omega + \varphi_{31}\varphi_{41}k+ \varphi_{32} (\varphi_{41} + \varphi_{42})
\end{eqnarray*}
and thus 
$$\varphi_3\bar{\varphi}_4+ \bar{\varphi}_3\varphi_4 = \varphi_{31}\varphi_{42} + \varphi_{32}\varphi_{41}.$$
Next, by (\ref{D2}) and simplifying, we  have $S_2=0$, where 
\begin{eqnarray}
	S_2&=&    \theta_1^3 \theta_{21} \theta_{31} k + \theta_1^3 \theta_{21} \theta_{32} k + \theta_1^3 \theta_{21} \theta_4 k + \theta_1^3 \theta_{22} \theta_{31} k + \theta_1^3 \theta_{31} \theta_4 + \theta_1^2 \theta_{21}^2 \theta_{31} + \notag \\ &&\theta_1^2 \theta_{21}^2 \theta_{32} k + \theta_1^2 \theta_{21}^2 \theta_{32} + \theta_1^2 \theta_{21}^2 \theta_4 k +
	\theta_1^2 \theta_{21} \theta_{22} \theta_{31} k + \theta_1^2 \theta_{21} \theta_{22} \theta_{31} + \theta_1^2 \theta_{21}  \theta_{22} \theta_{32}  \notag\\
	&& + \theta_1^2 \theta_{21} \theta_{22} \theta_4 + \theta_1 \theta_{21}^4 k + \theta_1 \theta_{21}^3 \theta_{22} + \theta_1 \theta_{21}^3 \theta_{31} k + \theta_1 \theta_{21}^3 \theta_{32} k + \theta_1 \theta_{21}^3 \theta_4 k + \notag\\
	&&\theta_1 \theta_{21}^2 \theta_{22}^2 + \theta_1 \theta_{21}^2 \theta_{22} \theta_{31} k + \theta_1 \theta_{21}^2 \theta_{31}^2 k + \theta_1 \theta_{21}^2 \theta_{31}^2 + \theta_1 \theta_{21}^2 \theta_{31} \theta_{32} + \notag\\
	&&\theta_1 \theta_{22}^2 \theta_{31}^2 + \theta_{21}^4 \theta_{32} k + \theta_{21}^4 \theta_4 k + \theta_{21}^4 \theta_4 +
	\theta_{21}^3 \theta_{22} \theta_{31} k + \theta_{21}^3 \theta_{22} \theta_{32} + \theta_{21}^3 \theta_{22} \theta_4. \label{S2}
\end{eqnarray}

In addition, from $\theta_2\bar{\theta}_2 + \theta_3\bar{\theta}_3 + \theta_1\theta_{4} + \theta_{4}^2=0,$
we have $S_3=0$, where
\begin{equation}
\label{case_2.2_L3}
S_3 = \theta_1 \theta_4 + \theta_{21}^2 k + \theta_{21} \theta_{22} + \theta_{22}^2 + \theta_{31}^2 k + \theta_{31} \theta_{32} + \theta_{32}^2 + \theta_4^2.
\end{equation}

The following lemma determines the condition on $(a_1,a_2,a_3)$ such that $\tr_{m}\left(\frac{\varphi_3\bar{\varphi}_3}{\varphi_1^2}\right)=0$, i.e., (\ref{D}) has solutions in this case.

\begin{Lemma}
	\label{caseD.2}
		If  $\varphi_1\neq0$,  $(\theta_1+ \theta_{21}) k +  \theta_{21} + \theta_{22} =0$, $S_1=S_2=S_3=0,$ where  $S_1,S_2,S_3 $ are defined as in (\ref{S1}), (\ref{S2}) and (\ref{case_2.2_L3}), respectively, then 
	$$T \triangleq \tr_{m}\left(\frac{\varphi_3\bar{\varphi}_3}{\varphi_1^2}\right) =0 $$if and only if 	$m$ is odd, $(\theta_{21}+\theta_{31})k^3+\theta_{22}k^2 + (\theta_{4}+\theta_{21})k + \theta_{21}+\theta_{22} =0$,  $(\theta_{21} + \theta_{31} ) k^2 + (\theta_{21}  + \theta_{22}  +  \theta_{31}) k + \theta_{32} = 0$ and $ (\theta_{21}^2  + \theta_{21} \theta_{31}) k + \theta_{21} \theta_{31} + \theta_{22} \theta_{31} =0.$
\end{Lemma}
\begin{proof} The proof is lengthy and is placed in first subsection of the appendix section.
\end{proof}

{\bfseries Case D.3} If $\varphi_1=0$ and $\varphi_2\neq0$, then $\theta_1 = \theta_2 + \bar{\theta}_2$.
From (\ref{D4}) and (\ref{D3}), we have $l_1^2 = \varphi_3\bar{\varphi}_3$ and $\theta_1l_1 = \theta_2\varphi_3 + \bar{\theta}_2\bar{\varphi}_3$, respectively and then by vanishing $l_1$, we obtain
$$\left( \theta_2^2 + \bar{\theta}_2^2 \right)  \varphi_3\bar{\varphi}_3 + \theta_2^2\varphi_3^2 + \bar{\theta}_2^2\bar{\varphi}_3^2 = 0, $$
i.e.,
\begin{equation}
\label{subcase2.3_eq1}
(\theta_2^2\varphi_3+\bar{\theta}_2^2\bar{\varphi}_3)(\theta_2+\bar{\theta}_2+\theta_3+\bar{\theta}_3) = 0.
\end{equation}
Moreover, from (\ref{D0}), we have $\tr_{m}\left( \frac{\varphi_4\bar{\varphi}_4}{\varphi_2^2} \right)=0$ and we assume $\frac{\varphi_4\bar{\varphi}_4}{\varphi_2^2} = \eta^2+\eta$, where $\eta\in\gf_{2^m}$. Then $l_3=\varphi_2\eta$. Furthermore, from (\ref{D1}), we get 
$$l_2= \frac{\theta_2\varphi_4+\bar{\theta}_2\bar{\varphi}_4}{\varphi_2} + \theta_1\eta.$$
Plugging the expressions of $l_1,l_2,l_3$ into (\ref{D2}) and simplifying (note that $\omega^2+\omega+k=0$), we get 
\begin{equation}
\label{subcase2.3_eq2}
(\theta_2^2\varphi_3+\bar{\theta}_2^2\bar{\varphi}_3)\left((\theta_2+\bar{\theta}_2) (\theta_2+\bar{\theta}_2+\theta_3+\bar{\theta}_3) (k^2 + k + \omega) + \theta_2\theta_3 + \bar{\theta}_2+\bar{\theta}_2\theta_3\right) =0.
\end{equation}

The following lemma tells that (\ref{D}) does not have solutions in this case. 

\begin{Lemma}\label{Lemma-3.6}
	Let $\theta_1=\theta_2+\bar{\theta}_2$, $\varphi_2\neq0$ and (\ref{subcase2.3_eq1}), (\ref{subcase2.3_eq2}) hold. Then 
	$$\tr_{m}\left(\frac{\varphi_4\bar{\varphi}_4}{\varphi_2^2}\right)=1.$$
\end{Lemma}
\begin{proof} The proof is lengthy and is placed in the second subsection of the appendix section.
\end{proof}

{\bfseries Case D.4} If $\varphi_1=\varphi_2=0$, then $\theta_1=\theta_2+\bar{\theta}_2$ and $\omega\theta_2=(\omega+1)\bar{\theta}_2$. If $\theta_1=0$, then $\theta_2=0$ and $f(x)$ is not an APN clearly by Proposition \ref{nu1=0}. Thus $\theta_1\neq0$, and then  (\ref{D}) has solutions in this case if and only if
\begin{equation}
\label{subcase2.4_eq1}
\theta_1^2\varphi_3\bar{\varphi}_3 = \theta_2^2\varphi_3^2 + \bar{\theta}_2^2 \bar{\varphi}_3^2,
\end{equation}
\begin{equation}
\label{subcase2.4_eq2}
\theta_1^2\varphi_4\bar{\varphi}_4 = \theta_2^2\varphi_4^2 + \bar{\theta}_2^2 \bar{\varphi}_4^2,
\end{equation}
and 
\begin{equation*}
\label{T}
\tr_{m}\left(\frac{\theta_2\bar{\theta}_2 + \varphi_3\bar{\varphi}_4 + \bar{\varphi}_3\varphi_4}{\theta_1^2}\right) =0. 
\end{equation*}

The following lemma characterizes the condition on $(a_1,a_2,a_3)$ such that (\ref{D}) has solutions in this case.
\begin{Lemma}
\label{caseD.4}
If  $\theta_1=\theta_2+\bar{\theta}_2\neq0$, $\omega\theta_2=(\omega+1)\bar{\theta}_2$, (\ref{subcase2.4_eq1}) and (\ref{subcase2.4_eq2}) hold, then
$$T \triangleq \tr_{m}\left(\frac{\theta_2\bar{\theta}_2 + \varphi_3\bar{\varphi}_4 + \bar{\varphi}_3\varphi_4}{\theta_1^2}\right) = 0$$ if and only if $m$ is odd and $\theta_{4}=\theta_{32}+\theta_{21}$. 
\end{Lemma}

\begin{proof}
We also assume that $\theta_2=\theta_{21}\omega+\theta_{22}$, $\theta_3 = \theta_{31}\omega+ \theta_{32}$, where $\theta_{21},\theta_{22},\theta_{31},\theta_{32}\in\gf_{2^m}$. Then from $\theta_1=\theta_2+\bar{\theta}_2$ and $\omega\theta_2=(\omega+1)\bar{\theta}_2$, we get 
$$\theta_1 = \theta_{21} = \theta_{22} \neq0. $$
Moreover, the expression of $T,$ (\ref{subcase2.4_eq1}), (\ref{subcase2.4_eq2}) and (\ref{theta_equation}) become
$$T =\tr_{m}\left(\frac{\theta_{31}(\theta_{21}+\theta_{31})k + \theta_{21}\theta_{32}+ \theta_{31}\theta_{32} + \theta_{32}^2 + \theta_{4}^2}{\theta_{21}^2}+k\right),$$
\begin{equation}
\label{subcase2.4_eq3}
(\theta_{21}+\theta_{31})((\theta_{21}+\theta_{31})k+\theta_{32}+\theta_{4}) =0,
\end{equation}
\begin{equation}
\label{subcase2.4_eq4}
((\theta_{21}+\theta_{31})k+\theta_{32}+\theta_{4})((\theta_{21}+\theta_{31})k^2+\theta_{31}k+\theta_{32})=0,
\end{equation}
and 
\begin{equation}
\label{subcase2.4_eq5}
(\theta_{21}+\theta_{31})^2k + \theta_{21}\theta_{4} + \theta_{31}\theta_{32} + \theta_{32}^2+\theta_{4}^2 =0.
\end{equation}

From (\ref{subcase2.4_eq3}) and (\ref{subcase2.4_eq4}), we divide the proof into two cases: (i) $\theta_{21}+\theta_{31}=0$ and $(\theta_{21}+\theta_{31})k^2+\theta_{31}k+\theta_{32}=0$; (ii) $(\theta_{21}+\theta_{31})k+\theta_{32}+\theta_{4}=0$. 

(i) In this case, $\theta_{21}=\theta_{31}$ and by (\ref{subcase2.4_eq5}), we have $\theta_{21}\theta_{4} + \theta_{31}\theta_{32} + \theta_{32}^2+\theta_{4}^2=0$, i.e., $\theta_{4} = \theta_{32}$ or $\theta_{4}=\theta_{32}+\theta_{21}$. If $\theta_{4} = \theta_{32}$, then
$$T=\tr_{m}(k)+\tr_{m}\left( \frac{\theta_{32}^2+\theta_{4}^2}{\theta_{21}^2} \right) = 1.$$
If $\theta_{4}=\theta_{32}+\theta_{21}$, then $$T=\tr_{m}(k)+\tr_{m}\left( \frac{\theta_{32}^2+\theta_{4}^2}{\theta_{21}^2} \right) = 1+\tr_{m}(1),$$
which equals zero if and only if $m$ is odd.

(ii) In this case, $(\theta_{21}+\theta_{31})k+\theta_{32}+\theta_{4}=0$. By (\ref{subcase2.4_eq5}), we get $$\theta_{21}\theta_{32}+\theta_{31}\theta_{4}+\theta_{32}^2+\theta_{4}^2=0.$$
Moreover,
\begin{eqnarray*}
	T&=& \tr_{m}(k) + \tr_{m}\left( \frac{\theta_{31}(\theta_{32}+\theta_{4})+\theta_{21}\theta_{32}+ \theta_{31}\theta_{32} + \theta_{32}^2 + \theta_{4}^2}{\theta_{21}^2} \right) \\
	&=& \tr_{m}(k) + \tr_{m}\left(\frac{\theta_{21}\theta_{32}+\theta_{31}\theta_{4}+\theta_{32}^2+\theta_{4}^2}{\theta_{21}^2}\right)\\
	&=& \tr_{m}(k) = 1.
\end{eqnarray*}
\end{proof}

\subsection{The proof of Step 5).}
 In this subsection, together with all conditions on $(a_1,a_2,a_3)$ in Propositions \ref{Prop_a=1}, \ref{four cases} and \eqref{con_case1}, we can finally finish the proof of Theorem \ref{theorem}. We also divide this part of the proof into three cases.  Note that we do not need to consider the case $\varphi_1=0$ and $\varphi_2\neq0$ since (\ref{D}) does not have solutions by Proposition \ref{four cases}  and thus $f(x)$ is not APN in this case. 
 
 {\bfseries Case E.1 $\varphi_1\varphi_2\neq0$.} In this case, from Proposition \ref{four cases}, (\ref{D}) has solutions if and only if 
 $ \theta_1\neq0, (\theta_1^2\theta_3+ \theta_1 \bar{\theta}_2^2 + \theta_2^2\theta_3+\bar{\theta}_2^2\bar{\theta}_3)(\theta_1^2\theta_4 + \theta_1\theta_2\bar{\theta}_2 + \theta_2^2\theta_3 + \bar{\theta}_2^2\bar{\theta}_3)=0, \tr_{m}\left( \frac{\varphi_3\bar{\varphi}_3}{\varphi_1^2} \right) = 0,$ under which the following two lemmas show the relation between $\tr_{m}\left(\frac{\theta_2\bar{\theta}_2}{\theta_1^2}\right)$ and $\tr_{m}\left( \frac{\varphi_3\bar{\varphi}_3}{\varphi_1^2} \right)$, as well as the reason why the conditions on $(a_1,a_2,a_3)$ are not the same when the odevity of $m$ is different  in Theorem \ref{theorem}. 
 
 \begin{Lemma}
 	\label{5.1_lemma1}
 	If $\varphi_1 \neq0$,  $\theta_1\neq0$ and $\theta_1^2\theta_4 + \theta_1\theta_2\bar{\theta}_2 + \theta_2^2\theta_3 + \bar{\theta}_2^2\bar{\theta}_3 = 0$. Then
 	$$\tr_{m}\left(\frac{\theta_2\bar{\theta}_2}{\theta_1^2}\right) = \tr_{m}\left( \frac{\varphi_3\bar{\varphi}_3}{\varphi_1^2} \right).$$
 \end{Lemma}
 
 \begin{proof}
 	Note that $\theta_2\bar{\theta}_2 + \theta_3\bar{\theta}_3 = \theta_1\theta_4 + \theta_4^2$. By computing directly, we have
 	\begin{eqnarray*}
 		&& \tr_{m}\left( \frac{\varphi_3\bar{\varphi}_3}{\varphi_1^2} + \frac{\theta_2\bar{\theta}_2}{\theta_1^2}  \right) \\
 		&=&\tr_{m}\left( \frac{ \theta_1^2\bar{\theta}_2\bar{\theta}_3 +  \theta_1^2\bar{\theta}_2\theta_4 + \theta_1^2\theta_2\theta_3 + \theta_1^2\theta_3\bar{\theta}_3 + \theta_1^2\theta_3\theta_4 + \theta_1^2\theta_2\theta_4 + \theta_1^2\bar{\theta}_3\theta_4+ \theta_2^3\bar{\theta}_2 + \theta_2\bar{\theta}_2^3 + \theta_1^2\theta_4^2 }{\theta_1^2\left( \theta_1^2 + \theta_2^2 + \bar{\theta}_2^2 \right)}  \right).
 	\end{eqnarray*}
 	Moreover, it is clear that 
 	\begin{eqnarray*}
 		&& \tr_{m}\left( \frac{\theta_1^2\theta_4^2 }{\theta_1^2\left( \theta_1^2 + \theta_2^2 + \bar{\theta}_2^2 \right)} \right)
 		= \tr_{m}\left( \frac{\theta_1^2\theta_4(\theta_1+\theta_2+\bar{\theta}_2)}{\theta_1^2\left( \theta_1^2 + \theta_2^2 + \bar{\theta}_2^2 \right)} \right) \\
 		&=& \tr_{m}\left(  \frac{\theta_1^2\theta_2\bar{\theta}_2 + \theta_1\theta_2^2\theta_3 + \theta_1\bar{\theta}^2\bar{\theta}_3 + \theta_1^2\theta_2\theta_4 + \theta_1^2\bar{\theta}_2\theta_4}{\theta_1^2\left( \theta_1^2 + \theta_2^2 + \bar{\theta}_2^2 \right)} \right),
 	\end{eqnarray*}
 	where the last equality is by $\theta_1^2\theta_4 = \theta_1\theta_2\bar{\theta}_2 + \theta_2^2\theta_3 + \bar{\theta}_2^2\bar{\theta}_3$. Thus 
 	\begin{eqnarray*}
 		&& \tr_{m}\left( \frac{(\bar{\theta}_2+\theta_3+\theta_4)({\theta}_2+\bar{\theta}_3+\theta_4)}{(\theta_1+\theta_2 + \bar{\theta}_2)^2} + \frac{\theta_2\bar{\theta}_2}{\theta_1^2}  \right) \\
 		&=& \tr_{m}\left( \frac{(\bar{\theta}_2\bar{\theta}_3 + \theta_2\theta_3) \theta_1(\theta_1+\theta_2+\bar{\theta}_2) + ( \theta_1^2\theta_4+\theta_1\theta_2\bar{\theta}_2)(\theta_3+\bar{\theta}_3)  + \theta_2^3\bar{\theta}_2 + \theta_2\bar{\theta}_2^3 + \theta_1^2(\theta_2\bar{\theta}_2+\theta_3\bar{\theta}_3) }{\theta_1^2\left( \theta_1^2 + \theta_2^2 + \bar{\theta}_2^2 \right)} \right)\\
 		&=&\tr_{m}\left( \frac{ \bar{\theta}_2^2\bar{\theta}_3^2 + \theta_2^2\theta_3^2 + (\theta_3+\bar{\theta}_3) (\theta_2^2\theta_3 + \bar{\theta}_2^2\bar{\theta}_3) + \theta_2^3\bar{\theta}_2 + \theta_2\bar{\theta}_2^3 + \theta_1^2(\theta_2\bar{\theta}_2+\theta_3\bar{\theta}_3)  }{\theta_1^2\left( \theta_1^2 + \theta_2^2 + \bar{\theta}_2^2 \right)} \right) \\
 		&=&\tr_{m}\left( \frac{ (\theta_2\bar{\theta}_2+\theta_3\bar{\theta}_3)\left( \theta_1^2 + \theta_2^2 + \bar{\theta}_2^2 \right)}{\theta_1^2\left( \theta_1^2 + \theta_2^2 + \bar{\theta}_2^2 \right)}\right)  \\
 		&=& \tr_{m}\left(\frac{\theta_2\bar{\theta}_2+\theta_3\bar{\theta}_3}{\theta_1^2}\right) =\tr_{m}\left(\frac{\theta_1\theta_4+\theta_4^2}{\theta_1^2}\right) =0. 
 	\end{eqnarray*}
 \end{proof}
 
 \begin{Lemma}
 	\label{5.1_lemma2}
 	If $\varphi_1 \neq0$,  $\theta_1\neq0$, $ \theta_1^2\theta_3+ \theta_1 \bar{\theta}_2^2 + \theta_2^2\theta_3+\bar{\theta}_2^2\bar{\theta}_3 = 0$ and $\theta_1^2\theta_4 + \theta_1\theta_2\bar{\theta}_2 + \theta_2^2\theta_3 + \bar{\theta}_2^2\bar{\theta}_3 \neq 0$. Then
 	
 	(1) $$\tr_{m}\left(\frac{\theta_2\bar{\theta}_2}{\theta_1^2}+\frac{\varphi_3\bar{\varphi}_3
 	}{\varphi_1^2}\right) = \tr_{m}(1).$$
 	
 	(2) Moreover, when $m$ is odd and $\tr_{m}\left(\frac{\varphi_3\bar{\varphi}_3
 	}{\varphi_1^2}\right)=0$, (\ref{con_case1}) does not hold. 
 \end{Lemma}
 
 \begin{proof}
 The proof is lengthy and placed in the third subsection in the appendix section.
 \end{proof}

 Lemma \ref{5.1_lemma2} tells that  if $\varphi_1 \neq0$,  $\theta_1\neq0$, $ \theta_1^2\theta_3+ \theta_1 \bar{\theta}_2^2 + \theta_2^2\theta_3+\bar{\theta}_2^2\bar{\theta}_3 = 0$, $\theta_1^2\theta_4 + \theta_1\theta_2\bar{\theta}_2 + \theta_2^2\theta_3 + \bar{\theta}_2^2\bar{\theta}_3 \neq 0$, $m$ is odd and $\tr_{m}\left(\frac{\varphi_3\bar{\varphi}_3
 }{\varphi_1^2}\right)=0$, (\ref{con_case1}) does not hold and thus $f(x)$ is not APN. Clearly, in another cases, the conditions $(a_1,a_2,a_3)$ in Proposition \ref{Prop_a=1} also hold and there is no $a\in\mu_{2^m+1}\backslash\{1\}$ such that $\nu_1=0$ by Proposition \ref{nu1=0}. Thus $f(x)$ is APN. 

{\bfseries Case E.2 $\varphi_1\neq0$ and $\varphi_2 = 0$.} In this case, by Proposition \ref{four cases}, (\ref{D}) has solutions if and only if $m$ is odd, $(\theta_{21}+\theta_{31})k^3+\theta_{22}k^2 + (\theta_{4}+\theta_{21})k + \theta_{21}+\theta_{22} =0$,  $(\theta_{21} + \theta_{31} ) k^2 + (\theta_{21}  + \theta_{22}  +  \theta_{31}) k + \theta_{32} = 0$ and $ (\theta_{21}^2  + \theta_{21} \theta_{31}) k + \theta_{21} \theta_{31} + \theta_{22} \theta_{31} =0.$  The following lemma tells that in this case, (\ref{con_case1}) does not hold and thus $f(x)$ is not APN. 
\begin{Lemma}
	If the conditions of Lemma \ref{caseD.2} hold, $m$ is odd, $(\theta_{21}+\theta_{31})k^3+\theta_{22}k^2 + (\theta_{4}+\theta_{21})k + \theta_{21}+\theta_{22} =0$,  $(\theta_{21} + \theta_{31} ) k^2 + (\theta_{21}  + \theta_{22}  +  \theta_{31}) k + \theta_{32} = 0$ and $ (\theta_{21}^2  + \theta_{21} \theta_{31}) k + \theta_{21} \theta_{31} + \theta_{22} \theta_{31} =0,$ then (\ref{con_case1}) does not hold.
\end{Lemma}

\begin{proof}
	It is clear that when $m$ is odd, we can let $k=1$ and $\omega$ satisfy $\omega^2+\omega+1=0$. Then in this case, we have $ \theta_{4} = \theta_{21} + \theta_{31}, \theta_{32}=\theta_{22}, \theta_1 = \theta_{22}, \theta_{21}^2 = \theta_{22} \theta_{31}$, under which we will show that $f(\bar{a}) = 1 + \omega^2 a_1 + \omega a_2 +a_3 =0$, which is a contradiction with (\ref{con_case1}).  We assume that $a_2=a_{21}\omega + a_{22}$ and $a_3 = a_{31}\omega + a_{32}$, where $a_{21}, a_{22}, a_{31}, a_{32} \in\gf_{2^m}.$ 
	Moreover, by computing, we have 
	$$ (1 + \omega^2 a_1 + \omega a_2 +a_3) (1+\omega a_1 + \omega^2 \bar{a}_2 + \bar{a}_3) = a_{21}+a_{22}+a_{31} + a_1(a_{22}+a_{31}+a_{32}).  $$
	Next, we show that  $a_{21}+a_{22}+a_{31} + a_1(a_{22}+a_{31}+a_{32}) $ always equals to zero. 
	From $\theta_{32}=\theta_{22}$, we have $$ (a_1+a_{21}) ( 1+ a_{31}+ a_{32})  = a_{22}(a_{32}+1). $$
	If $a_{32}=1$, then $a_{31}(a_1+a_{21})=0$. If $a_{31}=0$, then from $\theta_{4}+\theta_1+\theta_{22}+\theta_{21}+\theta_{31}=0$, we get $a_1+a_{21}+a_{22}=0$. Moreover, from $ \theta_{21}^2 = \theta_{22} \theta_{31},$ we have $a_{21}(a_1+a_{22}) =0$. Thus $a_{21}=0$ and $a_1=a_{22}$. Furthermore, we have $\theta_1+\theta_2+\bar{\theta}_2=0$, which is a contradiction. If $a_1=a_{21}$, then $a_{21}+a_{22}+a_{31} + a_1(a_{22}+a_{31}+a_{32}) = (a_{22}+a_{31})(a_1+1)$,  we assume which does not equal zero. Then from  $\theta_{4}+\theta_{21}+\theta_{31}=0$ and $\theta_1=\theta_{22}$, we get $a_{31}=1$ and $a_1=a_{22}$ and then $\theta_1+\theta_2+\bar{\theta}_2=0$, which is also a  contradiction.
	
	If $a_{32}\neq1$, then $a_{22} = \frac{(a_1+a_{21})(1+a_{31}+a_{32})}{a_{32}+1}$ and plugging it into $ \theta_{4} + \theta_{21} + \theta_{31} + \theta_1 + \theta_{22} = 0$, we have $(a_1+a_{21}+a_{32}+1)(a_{31}^2+a_{31}a_{32}+a_{32}^2+1) = 0$, which means $a_{31}^2+a_{31}a_{32}+a_{32}^2+1=0$, since if $a_1+a_{21}+a_{32}+1=0,$ then $a_{21}+a_{22}+a_{31} + a_1(a_{22}+a_{31}+a_{32}) = (a_1+a_{21}+a_{32}+1) (a_1a_{31}+a_1a_{32}+a_1+a_{31})=0$. Together with $a_{31}^2+a_{31}a_{32}+a_{32}^2+1=0$ and $\theta_{4}+\theta_{21}+\theta_{31}=0$, we obtain $a_{31}^2(a_1a_{31}+a_{21})=0$, meaning $\theta_1+\theta_2+\bar{\theta}_2=a_{31}^2(a_1a_{31}+a_{21})^2(a_1+a_{21}+a_{31}+1)^2=0$, which is a contradiction. Thus $1 + \omega^2 a_1 + \omega a_2 +a_3 = 0$ holds. 
\end{proof}

{\bfseries Case E.3 $\varphi_1=0$ and $\varphi_2=0$} In this case, by Proposition \ref{four cases}, (\ref{D}) has solutions if and only if $m$ is odd and $\theta_{4}=\theta_{32}+\theta_{21}$. The following lemma tells that in this case, $1+a_1+a_2+a_3=0$, which is a contradiction with Proposition \ref{Prop_a=1} and thus $f(x)$ is not APN. 

\begin{Lemma}
	If the conditions of Lemma \ref{caseD.4} hold and $m$ is odd, as well as $\theta_{4}=\theta_{32}+\theta_{21}$, then  $1+a_1+a_2+a_3=0$.
\end{Lemma}
\begin{proof}
	  It is clear that when $m$ is odd, we can let $k=1$ and $\omega$ satisfy $\omega^2+\omega+1=0$. Then we have $\theta_{32}=\theta_{31}k=\theta_{31}$.
	We assume that $a_2=a_{21}\omega + a_{22}$ and $a_3 = a_{31}\omega + a_{32}$, where $a_{21}, a_{22}, a_{31}, a_{32} \in\gf_{2^m}.$
	Moreover, by computing, we have 
	$$(1+a_1+a_2+a_3)(1+a_1+\bar{a}_2+\bar{a}_3) = (a_1+1)(a_{21}+a_{31}). $$
	Next, we assume $a_{21}\neq a_{31}$ and $a_1\neq1$. From  $\theta_{32}=\theta_{31}$, we get $a_{22}=a_1a_{32}$ and plugging it into  $\theta_{21}=\theta_{31}$, we have $(a_{32}+1)(a_1a_{31}+a_{21})=0$. If $a_{32}=1$, then from $\theta_1=\theta_{21}$, $\theta_{21}=\theta_{22}$, $\theta_{4}=0$, we obtain $(a_{21}+a_{31})(a_1+a_{21}+a_{31}+1)=0$, $a_{31}(a_1+a_{21})=0$ and $a_{21}(a_1+a_{21})=0$, respectively. Thus $a_1+a_{21}+a_{31}+1=0$ and $a_{21}=a_1$ (otherwise $a_{21} = a_{31} =0$), then $a_{31}=1$, leading to $\theta_{21}=0$, which is a contradiction. If $a_1a_{31}+a_{21}=0$, together with  $a_{22} = a_1a_{32}$, we have $a_2=a_1a_3$ and $\theta_1=(a_1^2+1)(1+a_{31}^2+a_{31}a_{32}+a_{32}^2+1)$. However, by $\theta_1=\theta_{21}$, we get   $1+a_{31}^2+a_{31}a_{32}+a_{32}^2+1=0$ and thus $\theta_1=0$, which is a contradiction. All in all, we have  $(a_1+1)(a_{21}+a_{31})=0$ and $1+a_1+a_2+a_3=0$.
\end{proof}

In conclusion, $f(x)$ is an APN function over $\gf_{2^{n}}$ if and only if 
\begin{enumerate}
	\item $m$ is even, $\theta_1\neq0, \varphi_1\varphi_2\neq0,$ $\tr_{m}\left(\frac{\theta_2\bar{\theta}_2}{\theta_1^2}\right)=0$ and $(\theta_1^2\theta_3+ \theta_1 \bar{\theta}_2^2 + \theta_2^2\theta_3+\bar{\theta}_2^2\bar{\theta}_3)(\theta_1^2\theta_4 + \theta_1\theta_2\bar{\theta}_2 + \theta_2^2\theta_3 + \bar{\theta}_2^2\bar{\theta}_3)=0$;
	\item $m$ is odd, $\theta_1\neq0, \varphi_1\varphi_2\neq0,$ $\tr_{m}\left(\frac{\theta_2\bar{\theta}_2}{\theta_1^2}\right)=0$ and $\theta_1^2\theta_4 + \theta_1\theta_2\bar{\theta}_2 + \theta_2^2\theta_3 + \bar{\theta}_2^2\bar{\theta}_3=0.$
\end{enumerate}
In addition, the condition $\varphi_1\varphi_2\neq0$ can be deleted. The reason is as follows.  
{If $\varphi_1=\theta_1+\theta_2+\bar{\theta}_2=0$, then $$\tr_{m}\left(\frac{\theta_2\bar{\theta}_2}{\theta_1^2}\right)=\tr_{m}\left( \frac{\theta_2\bar{\theta}_2}{\theta_2^2+\bar{\theta}_2^2} \right) = 1.$$ If $\varphi_2 = 0$, by \eqref{v2=0}, we have $\theta_1=\frac{\theta_{21}k+\theta_{21}+\theta_{22}}{k}.$ Let $u=\theta_{21}k+\theta_{21}+\theta_{22}$. Then $u\neq0$ since $\theta_1\neq0$ and  $$\tr_{m}\left(\frac{\theta_2\bar{\theta}_2}{\theta_1^2}\right) = \tr_{m}\left( \frac{\theta_{21}^2k^3 + \theta_{21}\theta_{22}k^2 + \theta_{22}^2k^2}{\left(\theta_{21}k+\theta_{21}+\theta_{22}\right)^2} \right)  = \tr_m\left( k^2 + \frac{\theta_{21}k^2}{u} + \frac{\theta_{21}^2k^4}{u^2}  \right) = 1.$$}

 Therefore, the proof of Theorem \ref{theorem} has finally been completed.

\section{Conclusion}

This paper presents a necessary and sufficient condition for a family of quadrinomial to be APN, which provides a complete answer to the first part of an open problem by Carlet \cite{carlet2014open}.
The main theorem is proved based on the Hasse-Weil bound, which is valid for all positive integers $m>5$. 
Experiment result indicates that the main theorem also holds for $m=4, 5$. 
Exhaustive search of APN quadrinomials of the discussed form gives two CCZ equivalent classes for $m=3$,
which are the Gold function $x^3$ and the Kim function, and gives two CCZ equivalent classes for $m=4$, which are the Gold functions $x^3$ and $x^9$. {For $m=5$ and $6$,  the numbers of APN functions $f(x)$ are $16336$ and $257858$, respectively.}
It indicates that the classification of APN quadrinomials $f(x)$ is beyond the computational capacity of a personal computer  when $m\ge 5$.  It is interesting to classify the APN quadrinomials for $m\geq 5$ in this paper and readers are cordially
invited to work on the (theoretical) classification of it under the CCZ equivalence.

\bibliographystyle{plain}

\bibliography{ref}

\section{Appendix}

\subsection{Proof of Lemma \ref{caseD.2}}
		After simplifying directly, we can have 
\begin{eqnarray*}
	T 
	&=& \tr_{m}\left( \frac{ (\theta_{21}^2+\theta_{31}^2)k + (\theta_{22}+\theta_{31}+\theta_{32}+\theta_{4})(\theta_{21}+\theta_{22}+\theta_{32}+\theta_{4})   }{(\theta_1+\theta_{21})^2}  \right).
\end{eqnarray*}
In addition,	plugging $\theta_1 = \frac{\theta_{21}+\theta_{22}}{k} + \theta_{21} $ into the expression of $T$, $S_1, S_2$ and $S_3$ and simplifying, we can obtain 
\begin{eqnarray*}
	T&=& \tr_{m}\left( \frac{(\theta_{21}^2+\theta_{31}^2)k^3 + (\theta_{22}+\theta_{31}+\theta_{32}+\theta_{4})(\theta_{21}+\theta_{22}+\theta_{32}+\theta_{4})k^2 }{(\theta_{21}+\theta_{22})^2}  \right),
\end{eqnarray*}
\begin{equation}
\label{L_1_th1}
\left((\theta_{21} + \theta_{31} ) k^2 + (\theta_{21}  + \theta_{22}  +  \theta_{31}) k + \theta_{32} \right) S_4 =0,
\end{equation}
\begin{equation}
\label{L_2_th1}
\left( (\theta_{21}^2  + \theta_{21} \theta_{31}) k + \theta_{21} \theta_{31} + \theta_{22} \theta_{31} \right) S_4 =0
\end{equation}
and 
\begin{equation}
\label{L_3_th1}
\left( 	\theta_{21}^2 + \theta_{31}^2 \right) k^2 + \left(  \theta_{21} \theta_{22}  +  \theta_{21} \theta_4  +\theta_{22}^2  + \theta_{31} \theta_{32} + \theta_{32}^2  + \theta_4^2 \right) k 
+ \theta_{21} \theta_4 + \theta_{22} \theta_4  =0,
\end{equation}
respectively, where
\begin{eqnarray}
S_4 &=& (\theta_{21}^3 + \theta_{21}^2 \theta_{31}) k^3 + (\theta_{21}^3+ \theta_{21}^2 \theta_{31}   + \theta_{21}^2 \theta_{32}  + \theta_{21} \theta_{22}^2 + \theta_{21}^2 \theta_4 + \theta_{22}^2 \theta_{31}  )  k^2 + \notag\\
&&	 (\theta_{21}^2 \theta_{22} +\theta_{22}^3)  k   + \theta_{21}^2 \theta_4     + \theta_{22}^2 \theta_4. 	\label{L_4}
\end{eqnarray} 
Next, from (\ref{L_1_th1}) and (\ref{L_2_th1}), there are two subcases here:  (i) $S_4=0$; (ii) $(\theta_{21} + \theta_{31} ) k^2 + (\theta_{21}  + \theta_{22}  +  \theta_{31}) k + \theta_{32} = 0$ and $ (\theta_{21}^2  + \theta_{21} \theta_{31}) k + \theta_{21} \theta_{31} + \theta_{22} \theta_{31} =0.$

(i) In this subcase, $S_4 = 0$.

If $\theta_{21}=0$, then $S_4 = \theta_{22}^2(\theta_{31}k^2 + \theta_{22} k + \theta_{4}) = 0$ and thus $   \theta_{4} = \theta_{31}k^2 + \theta_{22} k$ since $\theta_{22} = \theta_{22} +\theta_{21} =k(\theta_1+\theta_{21}) \neq0 $.  Moreover,  plugging $\theta_{21}=0$ and  $   \theta_{4} = \theta_{31}k^2 + \theta_{22} k$ into the expression of $T$ and (\ref{L_3_th1}) and simplifying, we get 
$$T = \tr_{m}\left( \frac{ \theta_{31}^2 k^6 + (\theta_{22}^2+\theta_{31}^2) k^4 + (\theta_{22}\theta_{31}+\theta_{31}^2)k^3 + (\theta_{22}^2+\theta_{22}\theta_{31} + \theta_{31}\theta_{32} + \theta_{32}^2)k^2 }{\theta_{22}^2} \right)$$
and
$$ k(\theta_{31}k^2 + (\theta_{22}+\theta_{31})k+\theta_{32})(\theta_{31}k^2 + (\theta_{22}+\theta_{31})k+\theta_{31}+\theta_{32}) =0 $$
respectively. Then $\theta_{31}k^2 + (\theta_{22}+\theta_{31})k + \theta_{32}=0$ or $\theta_{31}k^2 + (\theta_{22}+\theta_{31})k+\theta_{31}+\theta_{32} =0. $ Plugging them into the expression of $T$, we can both obtain  
$$T=\tr_{m}\left(\frac{\theta_{31}^2k^4 + \theta_{22}^2k^2 + \theta_{22}\theta_{31}k^2}{\theta_{22}^2}\right) = \tr_{m}\left(k^2 + \frac{\theta_{31}^2k^4}{\theta_{22}^2} + \frac{\theta_{31}k^2}{\theta_{22}} \right)=\tr_{m}(k^2)=1.$$ 

If $\theta_{21}\neq0$, then from $S_4 = 0$, we obtain $\theta_{32}=\frac{\Delta}{k^2\theta_{21}^2}$, where $\Delta
= (\theta_{21}^3 + \theta_{21}^2 \theta_{31}) k^3 + (\theta_{21}^3+ \theta_{21}^2 \theta_{31} + \theta_{21} \theta_{22}^2 + \theta_{21}^2 \theta_4 + \theta_{22}^2 \theta_{31}  )  k^2 +(\theta_{21}^2 \theta_{22} +\theta_{22}^3)  k   + \theta_{21}^2 \theta_4     + \theta_{22}^2 \theta_4$. Plugging it into the expression of $T$ and $S_3=0$ and simplifying, we get 
$$ T = \tr_{m}\left(\frac{S_5}{\theta_{21}^4(\theta_{21}+\theta_{22})^2k^4}\right),$$
and 
$$( (\theta_{21}+\theta_{31})k^2+\theta_{22}k+\theta_{4})S_6 =0,$$
where 
\begin{eqnarray*}
	S_5 &=& k^8 \theta_{21}^6 + k^8 \theta_{21}^4 \theta_{31}^2 + k^6 \theta_{21}^5 \theta_{22} + k^6 \theta_{21}^5 \theta_{31} + k^6 \theta_{21}^4 \theta_{22} \theta_{31} +k^6 \theta_{21}^2 \theta_{22}^4 + k^6 \theta_{21}^2 \theta_{22}^2 \theta_{31}^2 + \\
	&& k^6 \theta_{22}^4 \theta_{31}^2 + k^5 \theta_{21}^5 \theta_{22} + k^5 \theta_{21}^4 \theta_{22} \theta_{31} + k^5 \theta_{21}^3 \theta_{22}^3 + k^5 \theta_{21}^2 \theta_{22}^3 \theta_{31} + k^4 \theta_{21}^5 \theta_{4} + \\
	&& k^4 \theta_{21}^4 \theta_{22}^2 + k^4 \theta_{21}^4 \theta_{31} \theta_{4} + k^4 \theta_{21}^3 \theta_{22}^2 \theta_{4} + k^4 \theta_{21}^2 \theta_{22}^2 \theta_{31} \theta_{4} +
	k^4 \theta_{22}^6 + k^2 \theta_{21}^4 \theta_{4}^2 + k^2 \theta_{22}^4 \theta_{4}^2
\end{eqnarray*}
and
\begin{eqnarray*}
	S_6 &=& k^4 \theta_{21}^5 + k^4 \theta_{21}^4 \theta_{31} + k^3 \theta_{21}^5 + k^3 \theta_{21}^4 \theta_{22} + k^2 \theta_{21}^5 + k^2 \theta_{21}^2 \theta_{22}^2 \theta_{31} + k^2 \theta_{21} \theta_{22}^4 + \\
	&&k^2 \theta_{22}^4 \theta_{31} + k \theta_{21}^4 \theta_{22} + k \theta_{22}^5 +
	\theta_{21}^4 \theta_{4} + \theta_{22}^4 \theta_{4}.
\end{eqnarray*}

Therefore,  $\theta_{4} =(\theta_{21}+\theta_{31})k^2+\theta_{22}k$ or $S_6=0$. If the former holds, plugging it into the expression of $T$, we have 
\begin{eqnarray*}
	T &=& \tr_{m}\left( \frac{(\theta_{21}^2 + \theta_{31}^2)k^4 + (\theta_{21}+\theta_{22})(\theta_{22}+\theta_{31}) k^2}{(\theta_{21}+\theta_{22})^2} \right) \\
	&=& \tr_{m}\left( k^4 + \frac{(\theta_{22}^2 + \theta_{31}^2)k^4}{(\theta_{21}+\theta_{22})^2} + \frac{(\theta_{22} + \theta_{31})k^2}{\theta_{21}+\theta_{22}} \right) \\
	&=&\tr_{m}(k^4)=1.
\end{eqnarray*}
If the latter holds, i.e., $S_6=0$, plugging $S_6=0$ into the expression of $T$ to vanish $\theta_{4}$, we obtain 
\begin{eqnarray*}
	T &=& \tr_{m}\left( \frac{(\theta_{21}+\theta_{22})^2(\theta_{21}^2+\theta_{21}\theta_{22}+\theta_{22}\theta_{31})^2k^4+\theta_{22}(\theta_{21}+\theta_{22})^4(\theta_{21}+\theta_{22}+\theta_{31})k^2}{(\theta_{21}+\theta_{22})^6} \right. \\ &&\left.+\frac{\theta_{21}^4(\theta_{21}^2+\theta_{31}^2)k^6+\theta_{21}^2(\theta_{21}+\theta_{31})(\theta_{21}+\theta_{22})^3k^3}{(\theta_{21}+\theta_{22})^6} \right) \\
	&=& \tr_{m}\left(k^4\right) + \tr_{m}\left( \frac{\theta_{22}^2(\theta_{21}+\theta_{22})^2(\theta_{21}+\theta_{22}+\theta_{31})^2k^4 +\theta_{22}(\theta_{21}+\theta_{22})^4(\theta_{21}+\theta_{22}+\theta_{31})k^2}{(\theta_{21}+\theta_{22})^6}  \right) \\
	&&+ \tr_{m}\left(\frac{\theta_{21}^4(\theta_{21}+\theta_{31})^2k^6+\theta_{21}^2(\theta_{21}+\theta_{31})(\theta_{21}+\theta_{22})^3k^3}{(\theta_{21}+\theta_{22})^6}  \right)\\
	&=& \tr_{m}(k^4) =1.
\end{eqnarray*}

(ii) In this subcase,  $(\theta_{21} + \theta_{31} ) k^2 + (\theta_{21}  + \theta_{22}  +  \theta_{31}) k + \theta_{32} = 0$ and $ (\theta_{21}^2  + \theta_{21} \theta_{31}) k + \theta_{21} \theta_{31} + \theta_{22} \theta_{31} =0.$ Plugging $ \theta_{32} = (\theta_{21} + \theta_{31} ) k^2 + (\theta_{21}  + \theta_{22}  +  \theta_{31}) k $ into the expression of $T$ and (\ref{L_3_th1}) and simplifying, we obtain 
$$T = \tr_{m}\left( \frac{ (\theta_{21}^2+\theta_{31}^2) k^6 + \theta_{22}^2 k^4 + (\theta_{21}\theta_{22}+\theta_{22}\theta_{31})k^3 + (\theta_{22}+\theta_{31}+\theta_{4})(\theta_{21}+\theta_{22}+\theta_{4})k^2  }{(\theta_{21}+\theta_{22})^2} \right)$$
and
\begin{equation}
((\theta_{21}+\theta_{31})k^2 + \theta_{22} k + \theta_{4}) ((\theta_{21}+\theta_{31})k^3+\theta_{22}k^2 + (\theta_{4}+\theta_{21})k + \theta_{21}+\theta_{22}) =0,
\end{equation}
respectively. Then $(\theta_{21}+\theta_{31})k^2 + \theta_{22} k + \theta_{4}=0$ or  $(\theta_{21}+\theta_{31})k^3+\theta_{22}k^2 + (\theta_{4}+\theta_{21})k + \theta_{21}+\theta_{22} =0$. If the former holds, plugging it into the expression of $T$, we have 
$$T=\tr_{m}\left(\frac{(\theta_{21}^2+\theta_{31}^2)k^4+(\theta_{21}+\theta_{22})(\theta_{22}+\theta_{31})k^2}{(\theta_{21}+\theta_{22})^2}\right) = \tr_{m}(k^4)=1.$$
If the latter holds, then $\theta_{4} = ((\theta_{21}+\theta_{31})k^3+\theta_{22}k^2 + \theta_{21}k + \theta_{21}+\theta_{22})/k$ and plugging it into the expression of $T$, we get 
\begin{eqnarray*}
	T &=& \tr_{m}\left( \frac{(\theta_{21}^2+\theta_{31}^2)k^4+(\theta_{21}\theta_{22}+\theta_{22}^2+\theta_{22}\theta_{31})k^2 + (\theta_{21}+\theta_{22})(\theta_{21}+\theta_{31})k + \theta_{21}^2+\theta_{22}^2}{(\theta_{21}+\theta_{22})^2} \right) \\
	&=& \tr_{m}\left( \frac{(\theta_{21}^2+\theta_{31}^2)k^4+(\theta_{21}+\theta_{22})(\theta_{22}+\theta_{31})k^2}{(\theta_{21}+\theta_{22})^2} +1 \right) +\tr_{m}\left( \frac{\theta_{21}\theta_{31}k^2 + (\theta_{21}+\theta_{22})(\theta_{21}+\theta_{31})k}{(\theta_{21}+\theta_{22})^2} \right) \\
	&=& 1+\tr_{m}(1) +  \tr_{m}\left( \frac{(\theta_{21}\theta_{31}+\theta_{21}^2+\theta_{31}^2)k^2}{(\theta_{21}+\theta_{22})^2} \right).
\end{eqnarray*}
In addition, since $ (\theta_{21}^2  + \theta_{21} \theta_{31}) k + \theta_{21} \theta_{31} + \theta_{22} \theta_{31} =0$,  we have 
\begin{eqnarray*}
	& & \tr_{m}\left( \frac{(\theta_{21}\theta_{31}+\theta_{21}^2+\theta_{31}^2)k^2}{(\theta_{21}+\theta_{22})^2} \right) \\
	&=& \tr_{m}\left( \frac{(\theta_{21}\theta_{31}+\theta_{21}^2+\theta_{31}^2)\theta_{31}^2(\theta_{21}+\theta_{22})^2}{(\theta_{21}+\theta_{22})^2\theta_{21}^2(\theta_{21}+\theta_{31})^2} \right) \\
	&=& \tr_{m}\left( \frac{\theta_{21}\theta_{31}^3 + \theta_{31}^2(\theta_{21}+\theta_{31})^2}{\theta_{21}^2(\theta_{21}+\theta_{31})^2} \right)\\
	&=& \tr_{m}\left( \frac{\theta_{21}\theta_{31}^3+\theta_{21}\theta_{31}(\theta_{21}+\theta_{31})^2}{\theta_{21}^2(\theta_{21}+\theta_{31})^2} \right) \\
	&=& \tr_{m}\left( \frac{\theta_{21}\theta_{31}}{(\theta_{21}+\theta_{31})^2}\right) =0
\end{eqnarray*}
Thus in the case,
$$ T = 1+\tr_{m}(1) +0 = 1+ \tr_{m}(1)$$
and then $T=0$ if and only if $m$ is odd.

All in all, $T=0$ if and only if $m$ is odd, $(\theta_{21}+\theta_{31})k^3+\theta_{22}k^2 + (\theta_{4}+\theta_{21})k + \theta_{21}+\theta_{22} =0$,  $(\theta_{21} + \theta_{31} ) k^2 + (\theta_{21}  + \theta_{22}  +  \theta_{31}) k + \theta_{32} = 0$ and $ (\theta_{21}^2  + \theta_{21} \theta_{31}) k + \theta_{21} \theta_{31} + \theta_{22} \theta_{31} =0.$ 

\subsection{Proof of Lemma \ref{Lemma-3.6}}

	According to (\ref{subcase2.3_eq1}) and (\ref{subcase2.3_eq2}), we divide the proof into two cases: (i) $\theta_2^2\varphi_3+\bar{\theta}_2^2\bar{\varphi}_3=0$; (ii) $\theta_2+\bar{\theta}_2+\theta_3+\bar{\theta}_3 =0 $ and $\theta_2\theta_3 + \bar{\theta}_2+\bar{\theta}_2\theta_3 = 0$. 

Now we also let $\theta_2=\theta_{21}\omega+\theta_{22}$ and $\theta_3 = \theta_{31}\omega + \theta_{32}$. Then $$\varphi_2=\omega\theta_2+(\omega+1)\bar{\theta}_2=\theta_{21}+\theta_{22},$$
$\varphi_3=\varphi_{31}\omega+\varphi_{32}$ and $\varphi_4=\varphi_{41}\omega+\varphi_{42}$, where $\varphi_{31},\varphi_{32},\varphi_{41},\varphi_{42}$ are defined as in (\ref{p34}). Moreover, after simplifying, (\ref{theta_equation}), i.e., $\theta_2\bar{\theta}_2 + \theta_3\bar{\theta}_3 + \theta_1\theta_{4} + \theta_{4}^2 = 0$ is equivalent to
\begin{equation}
\label{subcases2.3_eq3}
(\theta_{21}+\theta_{31})^2k+\theta_{21}\theta_{22} + \theta_{22}^2 + \theta_{31} \theta_{32} + \theta_{32}^2 + \theta_{4} (\theta_{21}+\theta_{4}) =0. 
\end{equation}
In addition, 
\begin{eqnarray*}
	T \triangleq \tr_{m}\left(\frac{\varphi_4\bar{\varphi}_4}{\varphi_2^2}\right) 
	=\tr_{m}\left( \frac{\varphi_{41}^2k+\varphi_{41}\varphi_{42}+\varphi_{42}^2  }{\left(\theta_{21}+\theta_{22}\right)^2}  \right).
\end{eqnarray*}	

(i) In the case, $\theta_2^2\varphi_3+\bar{\theta}_2^2\bar{\varphi}_3=0$. Plugging $\theta_2=\theta_{21}\omega+\theta_{22}$ into the above equation, we get 
$$\theta_{21}^2\varphi_{31} k + \theta_{21}^2 \varphi_{31} + \theta_{21}^2\varphi_{32} + \theta_{22}^2\varphi_{31} =0, $$
i.e.,
\begin{equation}
\label{subcases2.3_eq4}
\theta_{21}^2(\theta_{21}+\theta_{31})k + \theta_{21}^2\theta_{22} + \theta_{21}^2 \theta_{31} + \theta_{21}^2\theta_{32} + \theta_{21}^2 \theta_{4} + \theta_{21}\theta_{22}^2 + \theta_{22}^2\theta_{31} =0. 
\end{equation}
From (\ref{subcases2.3_eq4}), if $\theta_{21}=0$, then $\theta_{22}^2\theta_{31}=0$ and $\theta_{31}=0$ since $\theta_{22} = \varphi_2+\theta_{21}\neq0$. Plugging $\theta_{21}=\theta_{31}=0$ into (\ref{subcases2.3_eq3}), we obtain $\theta_{4}=\theta_{22}+\theta_{32}$. Furthermore, $T$ can be reduced to 
\begin{eqnarray*}
	T=\tr_{m}\left( \frac{k\theta_{22}^2 + \theta_{22}\theta_{32} + \theta_{32}^2 }{\theta_{22}^2}\right) =\tr_{m}\left(k\right) =1.
\end{eqnarray*}
If $\theta_{21}\neq0$, from (\ref{subcases2.3_eq4}), we get $\theta_{4}=\frac{\Delta}{\theta_{21}^2}$, where $$\Delta=\theta_{21}^2(\theta_{21}+\theta_{31})k + \theta_{21}^2\theta_{22} + \theta_{21}^2 \theta_{31} + \theta_{21}^2\theta_{32} + \theta_{21}\theta_{22}^2 + \theta_{22}^2\theta_{31}.$$ 
Plugging $\theta_{4}=\frac{\Delta}{\theta_{21}^2}$ into (\ref{subcases2.3_eq3}) and the expression of $T$, we obtain 
\begin{equation}
\label{subcases2.3_eq5} 
(\theta_{21}+\theta_{31})\left(\theta_{21}^4 \left(\theta_{21}+\theta_{31}\right)k^2 + \theta_{21}^4\theta_{31}k + \theta_{21}^4\theta_{31} + \theta_{21}^4\theta_{32} + \theta_{21}^3\theta_{22}^2 + \theta_{21}\theta_{22}^4 + \theta_{22}^4\theta_{31} \right) =0
\end{equation}
and 
$$T = \tr_{m} \left( \frac{S_7}{\theta_{21}^4(\theta_{21}+\theta_{22})} \right)^2,$$
where  
\begin{eqnarray*}
	S_7 &=& k^4 \theta_{21}^6 + k^4 \theta_{21}^4 \theta_{31}^2 + k^2 \theta_{21}^6 + k^2 \theta_{21}^5 \theta_{22} + k^2 \theta_{21}^5 \theta_{31} + k^2 \theta_{21}^4 \theta_{22}^2 + k^2 \theta_{21}^4 \theta_{22} \theta_{31} + \\
	&& k^2 \theta_{21}^4 \theta_{31}^2 + k^2 \theta_{21}^2 \theta_{22}^4 +
	k^2 \theta_{21}^2 \theta_{22}^2 \theta_{31}^2 + k^2 \theta_{22}^4 \theta_{31}^2 + k \theta_{21}^5 \theta_{22} + k \theta_{21}^5 \theta_{31} + \\
	&& k \theta_{21}^4 \theta_{31}^2 + k \theta_{21}^3 \theta_{22}^3 + k \theta_{21}^2 \theta_{22}^3 \theta_{31} + k \theta_{21}^2 \theta_{22}^2 \theta_{31}^2 +\theta_{21}^4 \theta_{22} \theta_{32} + \\
	&&\theta_{21}^4 \theta_{31} \theta_{32} + \theta_{21}^4 \theta_{32}^2 + \theta_{21}^3 \theta_{22}^2 \theta_{32} +
	\theta_{21}^2 \theta_{22}^2 \theta_{31} \theta_{32}.
\end{eqnarray*}

From (\ref{subcases2.3_eq5}), we have $\theta_{21}+\theta_{31}=0$ or $\theta_{21}^4 \left(\theta_{21}+\theta_{31}\right)k^2 + \theta_{21}^4\theta_{31}k + \theta_{21}^4\theta_{31} + \theta_{21}^4\theta_{32} + \theta_{21}^3\theta_{22}^2 + \theta_{21}\theta_{22}^4 + \theta_{22}^4\theta_{31}=0$. If $\theta_{21}=\theta_{31}$, then 
\begin{eqnarray*}
	T &=& \tr_{m}\left( \frac{\theta_{21}^4\left( \theta_{21}^2k^2 + (\theta_{21}\theta_{22}+\theta_{22}^2)k+\theta_{21}\theta_{32}+\theta_{22}\theta_{32} + \theta_{32}^2   \right)}{\theta_{21}^4(\theta_{21}+\theta_{22})^2} \right) \\
	&=& \tr_{m}\left( k^2 + \frac{\theta_{22}^2k^2+\theta_{32}^2}{(\theta_{21}+\theta_{22})^2} + \frac{\theta_{22}k+\theta_{32}}{\theta_{21}+\theta_{22}} \right)\\
	&=&\tr_{m}(k) =1.
\end{eqnarray*}

If $\theta_{21}^4 \left(\theta_{21}+\theta_{31}\right)k^2 + \theta_{21}^4\theta_{31}k + \theta_{21}^4\theta_{31} + \theta_{21}^4\theta_{32} + \theta_{21}^3\theta_{22}^2 + \theta_{21}\theta_{22}^4 + \theta_{22}^4\theta_{31}=0$, then we can obtain an expression of $\theta_{32}$ and plugging it into $S_7$, we get 
\begin{eqnarray*}
	T&=& \tr_{m}\left( \frac{\left( \theta_{21}^8 + \theta_{21}^6\theta_{22}^2 + \theta_{21}^6\theta_{31}^2 + \theta_{21}^4\theta_{22}^2\theta_{31}^2 \right)k^2}{\theta_{21}^8} + \frac{\left(\theta_{21}^7\theta_{22} + \theta_{21}^7\theta_{31} + \theta_{21}^6\theta_{22}\theta_{31} \right) k}{\theta_{21}^8} \right. \\
	&&\left. + \frac{   \theta_{21}^6\theta_{22}\theta_{31} + \theta_{21}^5\theta_{22}^3 + \theta_{21}^4\theta_{22}^3+\theta_{31} + \theta_{21}^4 \theta_{22}^2 \theta_{31}^2 + \theta_{21}^2 \theta_{22}^6 + \theta_{22}^6\theta_{31}^2}{\theta_{21}^8}  \right) \\
	&=& \tr_{m}\left( k^2\right) + \tr_{m}\left( \frac{(\theta_{21}\theta_{22}+\theta_{21}\theta_{31}+\theta_{22}\theta_{31})^2k^2}{\theta_{21}^4} + \frac{(\theta_{21}\theta_{22}+\theta_{21}\theta_{31}+\theta_{22}\theta_{31})k}{\theta_{21}^2} \right) \\
	&& + \tr_{m}\left( \frac{\theta_{21}\theta_{22}^3+\theta_{22}^3\theta_{31}+ \theta_{21}^2 \theta_{22}\theta_{31} }{\theta_{21}^4} + \frac{(\theta_{21}\theta_{22}^3+\theta_{22}^3\theta_{31}+ \theta_{21}^2 \theta_{22}\theta_{31})^2 }{\theta_{21}^8} \right)\\
	&=& \tr_{m}(k^2) =1.
\end{eqnarray*} 
Thus when $\theta_2^2\varphi_3+\bar{\theta}_2^2\bar{\varphi}_3=0$, $\tr_{m}\left(\frac{\varphi_4\bar{\varphi}_4}{\varphi_2^2}\right)=1$. 

(ii) In this case, $\theta_2+\bar{\theta}_2+\theta_3+\bar{\theta}_3 =0,$ i.e., $\theta_{21}=\theta_{31}$, and $\theta_2\theta_3 + \bar{\theta}_2+\bar{\theta}_2\theta_3 = 0$. In addition, by computing, we know
$$ \theta_2\theta_3 + \bar{\theta}_2+\bar{\theta}_2\theta_3 = (\theta_{21}+\theta_{21}\theta_{31})\omega + \theta_{21}+\theta_{22}+\theta_{21}\theta_{32} =0 $$
and thus $\theta_{21}+\theta_{21}\theta_{31} = \theta_{21}^2+\theta_{21}=0$ and $\theta_{21}+\theta_{22}+\theta_{21}\theta_{32} =0$. If $\theta_{21}=0, $ then $\theta_{22} =0$, which means $\varphi_2=0$, which is a contradiction. Thus $\theta_{21}=1$ and $\theta_{22}+\theta_{32}=1$. Moreover, by (\ref{subcases2.3_eq3}), we get $\theta_{4}(\theta_{4}+1)=0$. If $\theta_{4}=0,$ then 
$$T = \tr_{m}\left(\frac{k^2+(\theta_{22}^2+\theta_{22})k}{(\theta_{22}+1)^2}\right) = \tr_{m}\left(\frac{k^2}{(\theta_{22}+1)^2}+\frac{k}{\theta_{22}+1}+k\right)=\tr_{m}(k)=1.$$
If $\theta_{4}=1$, then 
$$T = \tr_{m}\left( \frac{k\theta_{22}^2 + \theta_{22} + 1}{(\theta_{22}+1)^2} \right) = \tr_{m}\left( k + \frac{1}{(\theta_{22}+1)^2} + \frac{1}{\theta_{22}+1} \right) = \tr_{m}(k) = 1.$$

All in all, the proof has been finished.

\subsection{Proof of Lemma \ref{5.1_lemma2}}

	Firstly, from $ \theta_1^2\theta_3+ \theta_1 \bar{\theta}_2^2 + \theta_2^2\theta_3+\bar{\theta}_2^2\bar{\theta}_3 = 0$ and $\theta_1\neq0$, we have $$\theta_1\left( \theta_3+ \bar{\theta}_3 \right) = \theta_2^2 + \bar{\theta}_2^2, $$
and thus (i) $\theta_3=\bar{\theta}_3$ and $\theta_2=\bar{\theta}_2$; or (ii) $\theta_3\neq\bar{\theta}_3$ and $\theta_1=\frac{\theta_2^2 + \bar{\theta}_2^2}{ \theta_3+ \bar{\theta}_3}$. Next we divide the proof into two cases.

{\bfseries The proof of (1).}	(i) If $\theta_3=\bar{\theta}_3$ and $\theta_2=\bar{\theta}_2$, then from $ \theta_1^2\theta_3+ \theta_1 \bar{\theta}_2^2 + \theta_2^2\theta_3+\bar{\theta}_2^2\bar{\theta}_3 = 0$ and $\theta_1^2\theta_4 + \theta_1\theta_2\bar{\theta}_2 + \theta_2^2\theta_3 + \bar{\theta}_2^2\bar{\theta}_3 \neq 0$, we have $\theta_1\theta_3=\theta_2^2$ and $ \theta_3\neq \theta_4$, respectively. Moreover, together with $\theta_2\bar{\theta}_2+\theta_3\bar{\theta}_3=\theta_1\theta_4+\theta_4^2$ and $\theta_1\theta_3=\theta_2^2$, we have $ (\theta_1+\theta_3+\theta_4)(\theta_3 + \theta_4) =0  $ and thus $\theta_1=\theta_3+\theta_4$. Hence,
\begin{eqnarray*}
	& & \tr_{m}\left(\frac{\theta_2\bar{\theta}_2}{\theta_1^2}+\frac{\varphi_3\bar{\varphi}_3
	}{\varphi_1^2}\right) \\
	&=& \tr_{m}\left( \frac{\theta_2^2 + (\theta_2^2+\theta_3^2 + \theta_4^2)}{\theta_1^2}  \right) \\
	&=&\tr_{m}\left( \frac{\theta_1^2}{\theta_1^2} \right) = \tr_{m}(1).
\end{eqnarray*}

(ii) If $\theta_3\neq\bar{\theta}_3$ and $\theta_1=\frac{\theta_2^2 + \bar{\theta}_2^2}{ \theta_3+ \bar{\theta}_3}$, plugging the above equation into  $ \theta_1^2\theta_3+ \theta_1 \bar{\theta}_2^2 + \theta_2^2\theta_3+\bar{\theta}_2^2\bar{\theta}_3 = 0$ and simplifying, we obtain $$\left(\theta_2+\theta_3+\bar{\theta}_2+\bar{\theta}_3\right)^2\left( \theta_2^2\theta_3+\bar{\theta}_2^2\bar{\theta}_3 \right)=0.$$
If $\theta_2+\theta_3+\bar{\theta}_2+\bar{\theta}_3=0$, then $\theta_1 = \frac{\theta_2^2 + \bar{\theta}_2^2}{ \theta_3+ \bar{\theta}_3} = \theta_2+\bar{\theta}_2$, which is a contradiction with the condition $\theta_1+\theta_2+\bar{\theta}_2 \neq0$. Thus we have $\theta_2^2\theta_3=\bar{\theta}_2^2\bar{\theta}_3$. Moreover,  from $ \theta_1^2\theta_3+ \theta_1 \bar{\theta}_2^2 + \theta_2^2\theta_3+\bar{\theta}_2^2\bar{\theta}_3 = 0$ and $\theta_1^2\theta_4 + \theta_1\theta_2\bar{\theta}_2 + \theta_2^2\theta_3 + \bar{\theta}_2^2\bar{\theta}_3 \neq 0$, we have $\theta_1\theta_3=\bar{\theta}_2^2$ and $ \theta_3\bar{\theta}_3\neq\theta_4^2$, respectively. Next we use  $\bar{\theta}_2^2 = \theta_1\theta_3$ and  ${\theta}_2^2 = \theta_1\bar{\theta}_3$ to replace $\theta_2$ and $\bar{\theta}_2$. In addition, let $\theta_3+\bar{\theta}_3=\alpha$ and $\theta_3\bar{\theta}_3=\beta$. Then from $\theta_2^2\bar{\theta}_2^2+\theta_3^2\bar{\theta}_3^2=\theta_1^2\theta_4^2+\theta_4^4$, we have $ \theta_1^2\beta+\beta^2 = \theta_1^2\theta_4^2 + \theta_4^4 $, i.e., $\theta_1^2 = \theta_4^2+\beta$ since $\theta_4^2+\beta = \theta_4^2 +\theta_3\bar{\theta}_3 \neq0$. Moreover, we have
\begin{eqnarray*}
	& & \left( \bar{\theta}_2^2 + \theta_3^2 + \theta_4^2 \right) \left( \theta_2^2 + \bar{\theta}_3^2 + \theta_4^2 \right) \\
	&=& \left( \theta_1\theta_3+ \theta_3^2 + \theta_4^2   \right)\left(  \theta_1\bar{\theta}_3 + \bar{\theta}_3^2 + \theta_4^2 \right) \\
	&=& \theta_1^2\beta + \theta_1\theta_4^2\alpha + \beta^2 + \theta_1 \alpha \beta + \alpha^2\theta_4^2+\theta_4^4,
\end{eqnarray*}
and 
$$\theta_1^2+\theta_2^2+\bar{\theta}_2^2=\theta_1^2 + \theta_1(\theta_3+\bar{\theta}_3) = \theta_1^2 + \theta_1 \alpha.  $$
Besides, we have 
$$\tr_{m}\left(\frac{\theta_2\bar{\theta}_2}{\theta_1^2}\right) = \tr_{m}\left( \frac{\theta_2^2\bar{\theta}_2^2}{\theta_1^4} \right) = \tr_{m}\left(\frac{\theta_3\bar{\theta}_3}{\theta_1^2}\right)$$
and thus 
\begin{eqnarray*}
	T &\triangleq &\tr_{m}\left(\frac{\theta_2\bar{\theta}_2}{\theta_1^2}+\frac{\varphi_3\bar{\varphi}_3
	}{\varphi_1^2}\right) \\
	&=& \tr_{m}\left( \frac{\beta}{\theta_1^2} + \frac{ \theta_1^2 \beta + \theta_1\theta_4^2\alpha + \beta^2 + \theta_1\alpha\beta + \alpha^2\theta_4^2+\theta_4^4  }{\left( \theta_1^2 + \theta_1 \alpha \right)^2}               \right) \\
	&=& \tr_{m}\left( \frac{  \theta_4^4 + \theta_1\theta_4^2\alpha + \alpha^2\theta_4^2 + \alpha^2\beta + \beta^2 + \theta_1\alpha\beta   }{\left( \theta_1^2 + \theta_1 \alpha \right)^2}    \right). 
\end{eqnarray*}
Plugging $\theta_4^2 = \theta_1^2 + \beta$ into the above equation, we obtain 
\begin{eqnarray*}
	T &=& \tr_{m}\left(  \frac{\theta_1^4+\theta_1^3\alpha + \theta_1^2\alpha^2}{\left( \theta_1^2 + \theta_1 \alpha \right)^2} \right) \\
	&=& \tr_{m}\left( 1 + \frac{\theta_1\alpha}{\theta_1^2+\alpha^2} \right) =\tr_{m}(1).
\end{eqnarray*}

{\bfseries The proof of (2).} Clearly, when $m$ is odd and $\tr_{m}\left(\frac{\varphi_3\bar{\varphi}_3
}{\varphi_1^2}\right)=0$, $\tr_{m}\left(\frac{\theta_2\bar{\theta}_2}{\theta_1^2}\right)=1$ and by Proposition \ref{nu1=0}, there are two elements, e.g., $\bar{A}\in\gf_{2^{m}}$ such that $\nu_1=0$. Next, we will show that for such $\bar{A}$ and $\bar{a}^2=\frac{\bar{A}+\omega}{\bar{A}+\omega^2}$, $f(\bar{a})=0$. Clearly, now we can let $k=1$ and $\omega$ satisfy $\omega^2+\omega+1=0$. In the following, we also divide the proof into two cases.

(i) $\theta_3=\bar{\theta}_3$ and $\theta_2=\bar{\theta}_2$. In the case, we have $\theta_1=\theta_3+\theta_{4}$. Moreover, $\nu_1=0$ becomes $\theta_1A^2+\theta_1A+\theta_1+\theta_2=0$. Let $\frac{\theta_2+\theta_1}{\theta_1}=\xi^2+\xi$, where $\xi\in\gf_{2^{m}}$. Then $\bar{a}^2=\frac{\xi+\omega^2}{\xi+\omega}$ and simplifying $f(\bar{a})$, whose numerator equals to
\begin{eqnarray*}
	N(f(\bar{a})) &=& (1+a_1+a_2+a_3)\xi^3+ (\omega a_1+\omega^2 a_2 + \omega a_3 + \omega^2) \xi^2 \\
	&+& (\omega a_1 + \omega^2 a_2 + \omega^2 a_3 + \omega) \xi + \omega^2 a_1 + \omega a_2 + a_3 + 1. 
\end{eqnarray*} 
Plugging $\xi^2=\xi+\frac{\theta_2+\theta_1}{\theta_1}$ into $N(f(\bar{a}))$ and simplifying, we get 
$$N(f(\bar{a})) = N_1\xi+N_2,$$
where 
$$N_1 = (1+a_1+a_2+a_3)\theta_{22} + (a_3+1)\theta_1$$
and 
$$N_2 = (\omega^2a_1+\omega a_2 + \omega^2 a_3 + \omega)\theta_{22}+(\omega a_3+\omega^2)\theta_1.$$
We assume that $a_2=a_{21}\omega + a_{22}$ and $a_3 = a_{31}\omega + a_{32}$, where $a_{21}, a_{22}, a_{31}, a_{32} \in\gf_{2^m}.$ Then
\begin{equation}
\label{thetanew2}	
\left\{
\begin{array}{lr}
\theta_1 = 1+a_1^2+a_{22}^2+a_{21}a_{22}+a_{32}^2+a_{31}a_{32}+a_{21}^2+a_{31}^2  \\ 
\theta_2 =    (a_{21}a_{32}+a_{22}a_{31})\omega+  a_1+ a_{21}a_{32} + a_{22}a_{32} + a_{21}a_{31}\\
\theta_3 =  (a_{21}+a_1 a_{31})\omega+ a_{21}+a_{22}+a_1a_{31}+a_1a_{32} \\
\theta_4 = a_1^2+a_{22}^2+a_{21}a_{22}+a_{21}^2.
\end{array}
\right.
\end{equation}
From $\theta_3=\bar{\theta}_3$, we have $a_{21}=a_1a_{31}$ and then plugging it into $\theta_2=\bar{\theta}_2$, we get $a_{31}(a_1a_{32}+a_{22})=0$. If $a_{31}=0$, then by $ \theta_1=\theta_3+\theta_{4}$, we have $a_1a_{32}+a_{22}+a_{32}^2+1=0$. Moreover, $$N_1=(a_1+a_{22}+a_{32}+1)(a_1a_{32}+a_{22}+a_{32}^2+1)=0 $$ and 
$$N_2 = (a_1+\omega a_{22}+a_{32}+\omega)(a_1a_{32}+a_{22}+a_{32}^2+1)=0.$$
Thus $f(\bar{a})=0$. The subcase $ a_1a_{32}+a_{22} = 0$ is similar and we omit here.

(ii) In this case, we have $\theta_1\theta_{31}+\theta_{21}=0$, $\theta_{21}^2\theta_{32}=\theta_{22}^2\theta_{31}$, $\theta_1+\theta_{32}+\theta_{4}=0, \theta_1\theta_{31}+\theta_{21}^2=0$, $\theta_1\theta_{32}+\theta_{22}^2=0$, $\theta_1^2+\theta_{4}^2+\theta_{31}^2+\theta_{32}^2+\theta_{31}\theta_{32}=0$ and $\theta_{4}^2+\theta_{31}^2+\theta_{32}^2+\theta_{31}\theta_{32}\neq0$. Moreover, assume that $\bar{A}\in\gf_{2^{m}}$ satisfies $\varphi_1\bar{A}^2+\theta_1\bar{A} + \varphi_2=0$ and by computing and simplifying,  the numerator of $f(\bar{a})$ is 
$$N(f(\bar{a})) = N_1 \bar{A} + N_2,$$
where 
$$N_1 = (1+a_1+a_2+a_3)(\theta_1^2+\varphi_1\varphi_2)+(\omega a_1+\omega^2 a_2 + \omega a_3 + \omega^2)\theta_1\varphi_1+(\omega a_1+\omega^2 a_2 + \omega^2 a_3 + \omega)\varphi_1^2$$
and 
$$N_2 =(1+a_1+a_2+a_3)\theta_1\varphi_2+(\omega a_1+\omega^2 a_2 + \omega a_3 + \omega^2)\varphi_2\varphi_1+(\omega a_1+\omega^2 a_2 + \omega^2 a_3 + \omega)\varphi_1^2. $$
By $\theta_1= \theta_{32}+\theta_{4}$, we have $a_{21}=a_1a_{31}+a_1a_{32}+a_{22}+a_{31}^2+a_{31}a_{32}+a_{32}^2+1$ and plugging it into $\theta_1^2+\theta_{4}^2+\theta_{31}^2+\theta_{32}^2+\theta_{31}\theta_{32}=0$, we get $$(a_{22}+a_{1}a_{32})(a_{22}+a_{1}a_{32}+a_{31}^2+a_{31}a_{32}+a_{32}^2+1)=0.$$
If $a_{22}=a_{1}a_{32}$, then by $\theta_{21}^2\theta_{32}+\theta_{22}^2\theta_{31}=0,$ we obtain $$(a_{31}^2+a_{31}a_{32}+a_{32}^2+1)^3(a_1+a_{31})^2=0.$$
However, $\theta_{4}^2+\theta_{31}^2+\theta_{32}^2+\theta_{31}\theta_{32}\neq0$ tells that $a_{31}^2+a_{31}a_{32}+a_{32}^2+1\neq0$ and thus $a_1=a_{31}$. Then it is easy to check that  $N_1=N_2=0$ and thus $f(\bar{a})=0$. The subcase $a_{22}=a_{1}a_{32}+a_{31}^2+a_{31}a_{32}+a_{32}^2+1$ is similar and we omit it here.

\end{document}